\documentclass[a4paper,11pt]{amsart}
\usepackage{geometry}
\usepackage{amsmath}
\usepackage{amsthm}
\usepackage{amsfonts}
\usepackage{amssymb}
\usepackage{mathrsfs}
\usepackage{mathtools}
\usepackage{graphicx}
\usepackage[all]{xy}
\usepackage{tikz}
\usepackage{pgf}
\usepackage{pgffor}
\usetikzlibrary{arrows, positioning, matrix, decorations, decorations.pathreplacing,shapes.geometric,decorations.pathmorphing, shapes.misc, calc, fadings}
\usepackage{keyval}
\usepgfmodule{shapes}
\usepgfmodule{plot}
\usetikzlibrary{decorations}
\numberwithin{equation}{section}
\usepackage[hidelinks]{hyperref}
\usepackage{hhline}
\usepackage{multirow}
\usepackage{multicol}
\usepackage{todonotes}

\newtheorem{thm}{Theorem}
 \numberwithin{thm}{section}

\newtheorem{defi}[thm]{Definition}

\newtheorem{ej}[thm]{Example}
\newtheorem{lemma}[thm]{Lemma}
\newtheorem{remark}[thm]{Remark}
\newcommand{\R}{\mathbb{R}}

\newcommand{\N}{\mathbb{N}}
\newcommand{\Z}{\mathbb{Z}}

\DeclareMathOperator{\Id}{Id}
\DeclareMathOperator{\minor}{minor}

\newcommand{\arrowschem}[2]{\raisebox{-2ex}%
	{$\stackrel{\stackrel{\displaystyle#1}{\longrightarrow}}%
	{\stackrel{\longleftarrow}{#2}}$}}
	
\tikzset{acyl chain/.style={decorate,line width=1 pt,color=lightgray,decoration={random steps,segment length=1.5 pt,amplitude=0.5pt}}}
\tikzset{lipid head/.style={color=gray,fill=lightgray}}

\pgfdeclaredecoration{lipidleaflet}{initial}
{
 \state{initial}[width=\pgfdecoratedpathlength/floor(\pgfdecoratedpathlength/6pt)] 
  {
    \draw [acyl chain] (0 pt, -0.5 pt) -- (0 pt, -11 pt);
    \draw  [acyl chain] (2 pt, -0.5 pt) -- (2 pt, -11 pt);
    \filldraw [lipid head] (1 pt, -13 pt) circle (2.5 pt);
  }
  \state{final}
  {
    \pgfpathmoveto{\pgfpointdecoratedpathlast}
  }
}

\tikzset{
  terminal/.style={
    rounded rectangle,
    minimum size=3mm,
    very thick,draw=blue!80,
    top color=white,bottom color=black!20},
  skip loop/.style={to path={-- ++(0,#1) -| (\tikztotarget)}}
}

{
  \tikzset{terminal/.append style={text height=1ex,text depth=.25ex}}
  \tikzset{nonterminal/.append style={text height=1.5ex,text depth=.25ex}}
}

\begin{document}

\title{Regions of multistationarity in cascades of Goldbeter-Koshland loops}

\author{Magal\'{\i} Giaroli}
\address{Dto.\ de Matem\'atica, FCEN, Universidad de Buenos Aires, and IMAS (UBA-CONICET), Ciudad Universitaria, Pab.\ I, 
C1428EGA Buenos Aires, Argentina}
\email{mgiaroli@dm.uba.ar}

\author{Fr\'ed\'eric Bihan}
\address{Laboratoire de Math\'ematiques\\
         Universit\'e Savoie Mont Blanc\\
         73376 Le Bourget-du-Lac Cedex\\
         France}
\email{Frederic.Bihan@univ-savoie.fr}
\urladdr{http://www.lama.univ-savoie.fr/~bihan/}

\author{Alicia Dickenstein}
\address{Dto.\ de Matem\'atica, FCEN, Universidad de Buenos Aires, and IMAS (UBA-CONICET), Ciudad Universitaria, Pab.\ I, 
C1428EGA Buenos Aires, Argentina}
\email{alidick@dm.uba.ar}
\urladdr{http://mate.dm.uba.ar/~alidick}

\thanks{AD and MG are partially supported by UBACYT 20020100100242, 
CONICET PIP 11220150100473, and ANPCyT PICT 2013-1110, Argentina.}  
%\author{Fr\'ed\'eric Bihan, Alicia Dickenstein and Magal\'i Giaroli}

\date{}

\begin{abstract}
We consider cascades of enzymatic Goldbeter-Koshland loops~\cite{GK81} 
with any number $n$ of layers, for which there exist two layers involving the same phosphatase. 
Even if the number of variables and the number of conservation
laws grow linearly with $n$, we find explicit regions in {\it reaction rate constant} 
and {\it total conservation constant} space for which the associated mass-action kinetics 
dynamical system is multistationary. Our computations are based on the theoretical results of 
our companion paper~\cite{AGB1}, which are inspired by results in
real algebraic geometry by Bihan, Santos and Spaenlehauer~\cite{bihan}.
\end{abstract}

\keywords{Enzymatic cascades, Goldbeter-Koshland loops, sparse polynomial systems, multistationarity}
\maketitle

\section{Introduction}

Signal transduction is the process through which cells communicate with the external environment, 
interpret stimuli and respond to them. This mechanism is controlled by signaling cascades. 
Classical signaling pathways typically contain a cascade of phosphorylation cycles 
where the activated protein in one layer acts as the modifier enzyme in the next layer. 
An example of signaling cascades is the Ras cascade (see Figure~\ref{fig:RAS},
as it is usually depicted in the biochemistry literature), which is an
 important signaling pathway in mitogen-activated protein kinases (MAPKs). This cascade 
 reaction activates transcription factors and regulates gene expression. The Ras signaling 
 pathway has a significant role in the occurrence and development of diseases such as  cancer~\cite{HCC} or developmental defects~\cite{stan}.
One key property  is the occurrence of multistability, which triggers
different crucial cellular events. A basic condition for these different cellular responses
 is the emergence of multistationarity.

% Basics on chemical reaction networks and multistationarity
A reaction network $G$ on a given  set  of
$s$ chemical  species is a finite directed graph  
whose edges $\mathcal{R}$ represent the reactions and are labeled by parameters $\kappa\in \R^{|\mathcal{R}|}_{>0}$, 
 known as \textit{reaction rate constants}, and whose  vertices are labeled by complexes, usually represented
 as nonnegative integer linear combinations of species.
After numbering the species, a complex can be identified with a vector in $\Z^s_{\geq 0}$.
Under mass-action kinetics, $G$ defines the following autonomous system of ordinary differential equations
in the concentrations $x_1, x_2, \dots, x_s$ of the species as \textit{functions of time $t$}:
\begin{equation}\label{fam}\dot{x} = f(x) = \left(\frac{dx_1}{dt},\frac{dx_2}{dt},\dots,\frac{dx_s}{dt}\right) =
\sum_{y\rightarrow{y'}\in\mathcal{R}}\kappa_{yy'} \, x^{y} \, (y'-y),
\end{equation}
where $x=(x_1,x_2,\dots, x_s)$, $f=(f_1, \dots, f_s)$, $x^y = x_1^{y_1}x_2^{y_2}\dots x_s^{y_s}$ and $y\rightarrow{y'}$ indicates that 
 the complex $y$ reacts to the complex $y'$ and $(y,y')\in\mathcal{R}$.
The steady states of the system correspond to constant trajectories, that is, to the common zero set 
of the polynomials $f_1, \dots, f_s \in \R[x_1, \dots, x_s]$.
As the vector $\dot{x}(t)$ lies for all time $t$ 
in the linear subspace $S$ spanned by 
the reaction vectors $\{y'-y :  y\rightarrow{y'}\in \mathcal{R}\}$ (which is known as the \textit{stoichiometric subspace}),
it follows that any trajectory $x(t)$ 
lies in a translate of $S$. Moreover, if $x(0)=x^0\in\R^s_{> 0}$, then $x(t)$ lies for any $t$ (in
the domain of definition) in the \textit{stoichiometric 
	compatibility class} $(x^0+S)\cap\R^s_{\geqslant 0}$.  We will work with conservative systems and so
	all trajectories will be defined for any $t \ge 0$. The linear equations of $x^0+S$ give
 \textit{conservation laws}. If $x^0 \in \R^s_{> 0}$, we can also write  the linear variety $x^0+S$ in the form:
 $\{x \in \R^s \, : \, \ell_1(x)=T_1, \dots, \ell_{\sigma}(x)= T_{\sigma}\}$,
 where $\ell_1, \dots, \ell_{\sigma}$ are linear forms defining a basis of the subspace orthogonal to $S$ and
 $T=(T_1, \dots, T_\sigma) \in \R^\sigma_{\ge 0}$. These constant values are called {\it total conservation constants}.
 
The network $G$ is said to \textit{have the capacity for multistationarity} 
if there exists a choice of reaction rate constants $\kappa$ and total conservation constants $T$ 
such that there are two or more steady states of  system~\eqref{fam} in the 
stoichiometric compatibility class determined by $T$. 
Several articles studied the capacity for multistationarity from the structure of the directed
 graph of reactions~\cite{banajipantea,feliu, fw13,fc, anne,mueller, aliciaMer}, a line initiated in~\cite{cfI, cfII}. 
When the capacity for multistationarity of $G$ is determined, the following difficult 
step is to find values of  multistationary parameters. This is a question of
quantifier elimination in real algebraic geometry, which is effective,  but for 
interesting networks the complexity of the computations with general standard
tools is too high. Several articles  in the literature addressed this question, 
with different approaches based on ad-hoc computations,  injectivity results which use signs of minors in different forms, 
degree theory \cite{feliumincheva,cfr,mincheva,kfc,kothamanchu, sontag} and
the study of sparse real polynomials~\cite{gatermann}.

In this work,  we use tools from real algebraic geometry based on the papers~\cite{AGB1,bihan},
to analyze multistationarity in cascades of enzymatic Goldbeter-Koshland loops. 
A second important ingredient of our approach is the observation  
that enzymatic cascades have the structure of
\emph{MESSI systems}, introduced and studied in~\cite{aliciaMer}, from which an explicit
parametrization of the steady states can be obtained, even in presence of multistationarity.
We show how to deform a given set of parameters of the model to produce multistationarity, 
including both the reaction rate constants and the total concentration constants. Moreover,
we identify open sets where multistationarity occurs in the space of all these parameters.

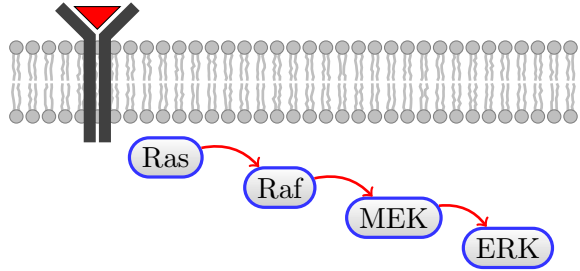
\begin{figure}[t]
\begin{tikzpicture}

% Lipid bilayer

\newcommand{\membranepath} { (0,0) .. controls (1,0) and (2,0) .. (7.5,0) } 
\draw [decorate, decoration={lipidleaflet}] \membranepath ; 
\draw [decorate, decoration = {lipidleaflet, mirror}] \membranepath ; 
\path \membranepath  node [sloped,near end, below=20pt] {};

\draw[line width=4.5pt,color=darkgray] (1,-0.8) -- (1,0.6) -- (0.6,1);
\draw[line width=4.5pt,color=darkgray] (1.2,-0.8) -- (1.2,0.6) -- (1.6,1);
\draw [fill=red] (0.8,1) -- (1.4,1) -- (1.1,0.7)--cycle;

\node [terminal] at (2.00000, -1.0000) (e0) {Ras};    
\node [terminal] at (3.500000, -1.4000) (e1) {Raf}; 
\node [terminal] at (5.00000, -1.800) (e2) {MEK};                     
\node [terminal] at (6.50000, -2.2000) (e3) {ERK}; 

\path[->,line width=1pt,color=red] (e0) edge [bend left] (e1);
\path[->,line width=1pt,color=red] (e1) edge [bend left] (e2);
\path[->,line width=1pt,color=red] (e2) edge [bend left] (e3);

\end{tikzpicture} 
\caption{The Ras pathway.} \label{fig:RAS}
\end{figure}

In Section~\ref{sec:1} we state and explain the theoretical setting presented in our companion paper~\cite{AGB1},
 based on the results in~\cite{bihan}. 
We apply our method to enzymatic cascades of Goldbeter-Koshland loops in Sections~\ref{sec:2} and \ref{sec:3}. 
In Section~\ref{sec:2} we apply our method to an enzymatic cascade with two layers 
and in Section~\ref{sec:3} we work with the general case of $n$ layers and present our main
results (Theorems~\ref{th:provisorio-cascaden} and ~\ref{th:cascadan-noconsecutivos}.)
In this case, the associated polynomial systems have positive dimensions growing linearly 
with $n$. The number of conservation relations 
(and then of total conservation constants) also grows linearly with $n$, and it is
at least four if $n \ge 2$. Such systems were studied in~\cite{sepulchereventura2, feliu2} when all the
enzymes are different, in which case there cannot be more than one positive steady state. 
This fact is proved in~\cite{feliu2} and also is a  particular case of a more general result
 in~\cite{treenetworks}, in which the authors work with a more general structure:
  tree networks of Goldbeter-Koshland loops. In the case of two layers, 
it was shown in~\cite{feliu} that in the case $n=2$ (see Figure~\ref{fig:sadi}),
 if the same phosphatase is acting at both layers, then the network has the capacity for multistationarity. 
It can be deduced from the results in~\cite{banajipantea}, that if there are any number of layers, and
the last two share a phosphatase, multistationarity parameters for the case $n=2$ can be 
extended to produce multistationarity parameters in the general case.

 Our approach allows us
to describe open sets in the rate constant and total concentration parameters which ensure 
multistationarity as long as at any pair of layers in the cascade there is a shared
phosphatase. Note for instance the structure of the statement of Theorem~\ref{th:provisorio-cascaden}:
if the given rate constants verify inequality~\eqref{eq:ineq}, then we give explicit inequalities
on the total conservation constants for which multistationarity occurs after tuning some 
 of the reaction rate constants not involved in~\eqref{eq:ineq}, that we clearly specify.

Our results can be generalized to describe multistationarity regions for other architectures of cascades which define
 MESSI systems.  For this purpose, we state and
prove in an Appendix the extension Theorem~\ref{thm:extending}, 
 that abstracts some of our computations in Section~\ref{sec:3}. 
We also refer to the general results in Section~5 in~\cite{AGB1}.
For example, in the case of the Ras cascade in Figure~\ref{fig:RAS}, previous papers
studied rate constant multistationarity parameters  (see e.g.~\cite{cfr,MerT}).
Our methods yield multistationarity regions for this signaling pathway in terms of rate constants and total
conservation parameters.  We omit these computations, which are similar
 to the ones we detail in Sections~\ref{sec:2} and~\ref{sec:3}.

\begin{figure}
	\centering
	\begin{tabular}{ccc}
		{\footnotesize (A)} \begin{tikzpicture}[scale=0.7,node distance=0.5cm]
		\node[] at (1.1,-1.5) (dummy2) {};
		\node[left=of dummy2] (p0) {$P_0$};
		\node[right=of p0] (p1) {$P_1$}
		edge[->, bend left=45] node[below] {\textcolor{black!70}{\small{$F$}}} (p0)
		edge[<-, bend right=45] node[above] {} (p0);
		\node[] at (-1,0) (dummy) {};
		\node[left=of dummy] (s0) {$S_0$};
		\node[right=of s0] (s1) {\textcolor{blue}{$S_1$}}
		edge[->, bend left=45] node[below] {\textcolor{black!70}{\small{$F$}}} (s0)
		edge[<-, bend right=45] node[above] {\textcolor{black!70}{\small{$E$}}} (s0)
		edge[->,thick,color=blue, bend left=25] node[above] {} ($(p0.north)+(25pt,10pt)$);
		\end{tikzpicture}
		& \hspace{0.6cm} &
		{\footnotesize (B)} \begin{tikzpicture}[scale=0.7,node distance=0.5cm]%%%
		\node[] at (1.1,-1.5) (dummy2) {};
		\node[left=of dummy2] (p0) {$P_0$};
		\node[right=of p0] (p1) {$P_1$}
		edge[->, bend left=45] node[below] {\textcolor{black!70}{\small{$F_2$}}} (p0)
		edge[<-, bend right=45] node[above] {} (p0);
		\node[] at (-1,0) (dummy) {};
		\node[left=of dummy] (s0) {$S_0$};
		\node[right=of s0] (s1) {\textcolor{blue}{$S_1$}}
		edge[->, bend left=45] node[below] {\textcolor{black!70}{\small{$F_1$}}} (s0)
		edge[<-, bend right=45] node[above] {\textcolor{black!70}{\small{$E$}}} (s0)
		edge[->,thick,color=blue, bend left=25] node[above] {} ($(p0.north)+(25pt,10pt)$);
		\end{tikzpicture}
	\end{tabular}
	\caption{Same and different phosphatases in a 2-layer cascade of GK-loops.}\label{fig:sadi}
\end{figure}
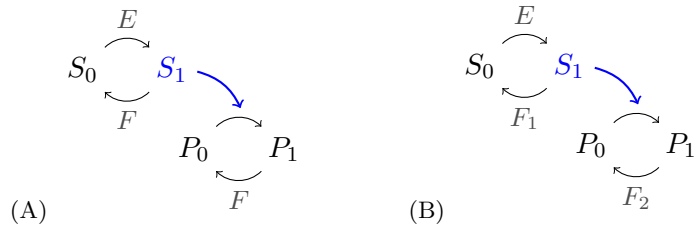

 \section{Positive solutions of sparse polynomial systems}\label{sec:1}
 
 We summarize some results from Section~2 in~\cite{AGB1}, where complete
 details can be found. We also refer the reader to~\cite{triangulations} for the
 combinatorial objects that we introduce in this section.
 
 We consider a  \textit{polynomial system} of $d$ Laurent polynomial equations 
 $f_{1}(x)=\dots=f_{d}(x)=0$ in $d$ variables $x=(x_1, \dots, x_d)$, with
 \begin{equation}\label{systemf}
 f_{i}(x)=\sum_{j=1}^n c_{ij} \, x^{a_j} \in \R[x_1,\dots,x_d], \ i=1,\dots,d,
 \end{equation}
 where the exponents belong to a fixed finite point configuration $\mathcal{A} =
  \{a_1, \dots, a_n\}\subset \Z^d$, with $n+2\geq d$.   We denote by $C = (c_{ij}) \in \R^{d \times n}$ the 
  \textit{coefficient matrix} of the system and we assume w.l.o.g.
   that no column of $C$ is identically zero. The set  $\mathcal{A}$ is 
  called the \textit{support} of the polynomial system.

In order to understand the notation we present a basic example.
 
\begin{ej}\label{example:support} \rm  We consider the following point configuration
	\begin{equation}\label{eq:support}
	\mathcal{A}=\{(0,0),(2,0),(0,1),(2,1),(1,2),(1,3)\}\subset \Z^2,
	\end{equation}
	and the coefficient matrix
	\begin{equation}\label{eq:coefficientmatrix}
		C=\left(\!\!\begin{array} {rrrrrr}
		1 & -2 & 1 & 1 & -1 & 0\\
	-2 & 1 & 0 & -1 & -1 & 1
	\end{array}\right).
	\end{equation}
	The polynomial system of two polynomial equations and two variables $x, y$:
	$$\begin{array}{rcl}
	1 - 2 \, x^2 + y + x^2 y - x y^2 & = & 0,\\
	-2 + x^2 - x^2 y - x y^2 + x y^3 & = & 0,
	\end{array}$$
	has support $\mathcal{A}$ and coefficient matrix $C$, and we write it in
	the form:
	\[ C \begin{pmatrix} 1 & x^2 & y & x^2y & xy^2 & xy^3
	\end{pmatrix}^t=0, \]
	where $t$ denotes the transpose.
\end{ej}

Our idea to ensure multistationarity, based on \cite{bihan}, is to restrict our polynomial 
system \eqref{systemf} to subsystems which have a positive solution and then extend 
these solutions to the total system, under a deformation of the coefficients. 
So, we are first interested in finding conditions  in the coefficient matrix that guarantee a positive solution in the subsystems. 

Suppose that the convex hull of $\mathcal{A}$ is a full-dimensional polytope $Q$. 
Figure~\ref{fig:support} shows the convex hull of $\mathcal{A}$ of Example~\ref{example:support}.
 A \textit{$d$-simplex} with vertices in $\mathcal{A}$ is a subset of $d+1$ points of $\mathcal{A}$
  affinely independent.
 Following Section 3 in~\cite{bihan}, we define:

 \begin{defi} Given any $d\times (d+1)$ matrix $M$ with real entries, we denote
 by  $\minor(M,i)$ the determinant of the square matrix obtained by
  removing the $i$-th column. The matrix $M$ is called positively spanning 
 if all the values $(-1)^i\minor(M,i)$, for $i=1,\dots,d+1$, are nonzero and have the same sign.
 \end{defi}
 
Thus, a matrix is positively spanning if all the coordinates of any non-zero vector 
 in the kernel of the matrix are non-zero and have the same sign. For example, the matrix
 \[M=\begin{pmatrix} 1 & 0 &-1\\
 0 & 1 & -1
 \end{pmatrix}
 \]
 is positively spanning because the values $(-1)^1\minor(M,1)=-1$, $(-1)^2\minor(M,2)=-1$ 
 and $(-1)^3\minor(M,3)=-1$ are all nonzero and have the same sign. Also, any non-zero vector 
 in the kernel of $M$ is of the form $\lambda(1,1,1)$, with $\lambda\in\R-\{0\}$, that is, 
 all the coordinates of one of these vectors share the same sign. 
 
 \begin{figure}[t]
	\centering
	\begin{tikzpicture}
	[scale=1.200000,
	back/.style={loosely dotted, thin},
	edge/.style={color=black, thick},
	facet/.style={fill=red!70!white,fill opacity=0.800000},
	vertex/.style={inner sep=1.5pt,circle,draw=black,fill=black,thick,anchor=base
	}]
	%
	%
	%% Coordinate of the vertices:
	%%
	\coordinate (0.00000, 0.00000) at (0.00000, 0.00000);
	\coordinate (0.00000, 1.00000) at (0.00000, 1.00000);
	\coordinate (2.00000, 0.00000) at (2.00000, 0.00000);
	\coordinate (2.00000, 1.00000) at (2.00000, 1.00000);
	\coordinate (1.00000, 2.00000) at (1.00000, 2.00000);
	\coordinate (1.00000, 3.00000) at (1.00000, 3.00000);
	%%
	%%
	%% Drawing the interior
	%%
	\fill[facet, fill=yellow!40!white,fill opacity=0.800000] (0.0000, 0.00000) --(0.0000, 1.00000) --  (1.00000, 3.00000)-- (2.00000,1.00000) -- (2.00000,0.00000) -- cycle {};

	%%
	%%
	%% Drawing edges
	%%
	\draw[edge] (0.00000, 0.00000) -- (0.00000, 1.00000);
	\draw[edge] (0.00000, 0.00000) -- (2.00000, 0.00000);
	\draw[edge] (2.00000, 0.00000) -- (2.00000, 1.00000);
	\draw[edge] (2.00000, 1.00000) -- (1.00000, 3.00000);
	\draw[edge] (1.00000,3.00000) -- (0.00000, 1.00000);
	%%
	%%
	%% Drawing the vertices
	%%
	\node[vertex] at (0.00000, 0.00000){};
	\node[vertex] at (0.00000, 1.00000){};
	\node[vertex] at (2.00000, 0.00000){};
	\node[vertex] at (2.00000, 1.00000){};
	\node[vertex] at (1.00000, 2.00000){};
	\node[vertex] at (1.00000, 3.00000){};
	\end{tikzpicture}
	\caption{Convex hull of the support $\mathcal{A}$ of the Example~\ref{example:support}.}
	\label{fig:support}
\end{figure}
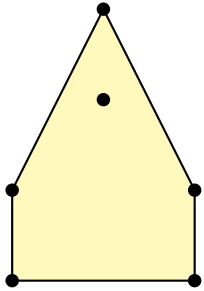
 
 It can be shown (Proposition~3.3 in \cite{bihan}) that if a polynomial system of $d$ polynomial 
 equations in $d$ variables has a $d$-simplex as support, then it has one non-degenerate positive solution
  if and only if its $d\times (d+1)$ matrix of coefficients is positively spanning.
  
 \begin{defi} \label{def:dec}
 Let $C$ be a $d\times n$ matrix with real entries. A $d$-simplex 
 $\Delta=\{a_{i_1},\dots,a_{i_{d+1}}\}$ is said to be 
 positively decorated by $C$
  if the $d\times(d+1)$ submatrix of $C$ with columns $\{i_1,\dots,i_{d+1}\}$ is positively spanning.
 \end{defi}

\begin{ej}\label{example:simplices} \rm (Example \ref{example:support} continued).
Consider again the support with vertices in $\mathcal{A}$ in~\eqref{eq:support}, and the following $2$-simplices:
\begin{align*}
\Delta_1=\{(0,0),(2,0),(0,1)\},\quad \Delta_2=\{(2,0), (0,1), (2,1)\},\quad \Delta_3=\{(0,1),(2,1),(1,2)\},
\end{align*}
depicted in Figure~\ref{fig:simplices}. 
The submatrix of $C$ given by the columns of $\Delta_1$ (the first three columns) equals:
 \[\left(\!\!\begin{array}{rrr}
  1 & -2 &1\\
-2 & 1 & 0
\end{array}\right).
\]
This matrix is positively spanning, so the simplex $\Delta_1$ is positively decorated by $C$.
 It is easy to check that the simplex $\Delta_2$ is also positively decorated by $C$.
%The submatrix of $C$ given by the columns of $\Delta_2$ (columns 2, 3 and 4) equals
%\left(\begin{array}{rrr}
% -2 &1 &1\\
%1 & 0 & -1
%\end{array}\right).
%\]
%This matrix is also positively spanning, so the simplex $\Delta_2$ is positively decorated by $C$ too. 
But the submatrix  given by the columns of $\Delta_3$ (the columns 3, 4 and 5):
\[\left(\begin{array}{rrr}
  1 &1 & -1\\
0 & -1 & -1
\end{array}\right)
\]
is not positively spanning, and then the simplex $\Delta_3$ is not positively decorated by $C$. 
	
\end{ej}

Given a $d$-simplex $\Delta$ with vertices in $\mathcal{A}$, we consider  height vectors $h\in\R^n$,
 where each coordinate $h_j$ of $h$ gives the value of a lifting function on
 the point $a_j$ of $\mathcal{A}$. 
 Denote by $\varphi_{\Delta,h}$ the unique affine function that agrees with $h$ on the points of $\Delta$,
 that is, $\varphi_{\Delta,h}(a_j)=h_j$ for all $a_j\in\Delta$. 
 We associate with $\Delta$ the following cone:
\[\mathcal{C}_{\Delta}=\{h=(h_1,\dots,h_n)\in\R^n \, : \,  \varphi_{\Delta,h}(a_j) \, < \, h_j \text{ for all } a_j\notin \Delta\}. \]

Assume that two simplices $\Delta_1$ and $\Delta_2$ share a common facet, that is, the points
in $\mathcal A$ that lie in a face of dimension $d-1$ of its convex hull,
and they only intersect there (see Figure~\ref{fig:2simplicescommonfacet} for an example). 
In this case, it can be shown that the cone $\mathcal{C}_{\Delta_1,\Delta_2}$ defined as the intersection
\begin{equation}\label{eq:Delta12}
\mathcal{C}_{\Delta_1,\Delta_2}=\mathcal{C}_{\Delta_1}\cap \mathcal{C}_{\Delta_2},\end{equation}
 is nonempty (see Proposition~2.5 in~\cite{AGB1}).

  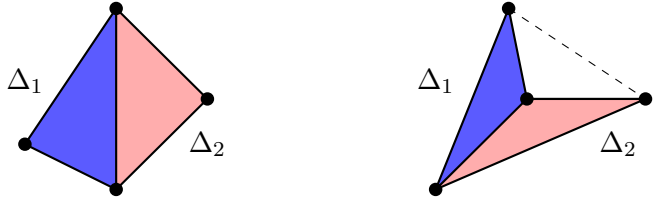
\begin{figure}[h]
	\centering
	\begin{tikzpicture}
	[scale=1.200000,
	back/.style={loosely dotted, thin},
	edge/.style={color=black, thick},
	facet/.style={fill=red!70!white,fill opacity=0.800000},
	vertex/.style={inner sep=1.5pt,circle,draw=black,fill=black,thick,anchor=base
	}]
	%
	%
	%% Coordinate of the vertices:
	%%
	\coordinate (0.00000, 0.50000) at (0.00000, 0.50000);
	\coordinate (2.00000, 1.00000) at (2.00000, 1.00000);
	\coordinate (1.00000, 0.00000) at (1.00000, 0.00000);
	\coordinate (1.00000, 2.00000) at (1.00000, 2.00000);
	%%
	%%
	%% Drawing the interior
	%%
	\fill[facet, fill=blue!80!white,fill opacity=0.800000] (1.00000, 2.00000) --  (1.00000, 0.00000)-- (0.00000,0.50000) -- cycle {};
	\fill[facet, fill=red!40!white,fill opacity=0.800000] (1.00000, 2.00000) --  (1.00000, 0.00000)-- (2.00000,1.00000) -- cycle {};
	
	%%
	%%
	%% Drawing edges
	%%
	\draw[edge] (0.00000, 0.50000) -- (1.00000, 2.00000);
	\draw[edge] (0.00000, 0.50000) -- (1.00000, 0.00000);
	\draw[edge] (1.00000, 0.00000) -- (1.00000, 2.00000);
	\draw[edge] (1.00000, 0.00000) -- (2.00000, 1.00000);
	\draw[edge] (1.00000, 2.00000) -- (2.00000, 1.00000);
	%%
	%%
	%% Drawing the vertices
	%%
	\node[vertex] at (0.00000, 0.50000){};
	\node[vertex] at (1.00000, 0.00000){};
	\node[vertex] at (1.00000, 2.00000){};
	\node[vertex] at (2.00000, 1.00000){};
	\node at (2.00000, 0.50000){$\Delta_2$};
	\node at (0.00000, 1.20000){$\Delta_1$};
	\end{tikzpicture}
	\hspace{2cm}
	\begin{tikzpicture}
	[scale=1.200000,
	back/.style={loosely dotted, thin},
	edge/.style={color=black, thick},
	facet/.style={fill=red!70!white,fill opacity=0.800000},
	vertex/.style={inner sep=1.5pt,circle,draw=black,fill=black,thick,anchor=base
	}]
	%
	%
	%% Coordinate of the vertices:
	%%
	\coordinate (0.00000, 0.00000) at (0.00000, 0.00000);
	\coordinate (1.00000, 1.00000) at (1.00000, 1.00000);
	\coordinate (2.30000, 1.00000) at (2.30000, 1.00000);
	\coordinate (0.80000, 2.00000) at (0.80000, 2.00000);
	%%
	%%
	%% Drawing the interior
	%%
	\fill[facet, fill=blue!80!white,fill opacity=0.800000] (0.80000, 2.00000) --  (1.00000, 1.00000)-- (0.00000,0.00000) -- cycle {};
	\fill[facet, fill=red!40!white,fill opacity=0.800000] (1.00000, 1.00000) --  (0.00000, 0.00000)-- (2.30000,1.00000) -- cycle {};
	
	%%
	%%
	%% Drawing edges
	%%
	\draw[edge] (0.00000, 0.00000) -- (1.00000, 1.00000);
	\draw[edge] (0.00000, 0.00000) -- (0.80000, 2.00000);
	\draw[edge] (0.00000, 0.00000) -- (2.30000, 1.00000);
	\draw[edge] (0.80000, 2.00000) -- (1.00000, 1.00000);
	\draw[edge] (1.00000, 1.00000) -- (2.30000, 1.00000);
	\draw[edge] (1.00000, 1.00000) -- (2.30000, 1.00000);
	\draw [dashed] (0.80000, 2.00000) -- (2.30000, 1.00000);
	%%
	%%
	%% Drawing the vertices
	%%
	\node[vertex] at (0.00000, 0.00000){};
	\node[vertex] at (1.00000, 1.00000){};
	\node[vertex] at (0.80000, 2.00000){};
	\node[vertex] at (2.30000, 1.00000){};
	\node at (2.00000, 0.50000){$\Delta_2$};
	\node at (0.00000, 1.20000){$\Delta_1$};
	\end{tikzpicture}
	\caption{Examples of simplices $\Delta_1$ and $\Delta_2$, which share a facet, $d=2$.} \label{fig:2simplicescommonfacet}
\end{figure}

\begin{ej}\label{example:cone} \rm (Example \ref{example:support} and \ref{example:simplices} 
continued). We compute the cone $\mathcal{C}_{\Delta_1,\Delta_2}$ for the simplices 
$\Delta_1$ and $\Delta_2$. We take any values $h_1$, $h_2$, $h_3$ 
corresponding to the points of $\Delta_1$: $(0,0)$, $(2,0)$ and $(0,1)$ respectively.

We consider then the unique affine linear function $\varphi_{\Delta_1,h}(x,y)$ which satisfies 
$\varphi_{\Delta_1,h}(0,0)=h_1$, $\varphi_{\Delta_1,h}(2,0)=h_2$ and $\varphi_{\Delta_1,h}(0,1)=h_3$:
 \[\varphi_{\Delta_1,h}(x,y)=\left(\frac{h_2-h_1}{2}\right)x + (h_3-h_1)\, y + h_1.\]
We need that
\begin{align*}
\varphi_{\Delta_1,h}(2,1)&=-h_1+h_2+h_3<h_4, \quad \varphi_{\Delta_1,h}(1,2)=-\frac{3}{2}\, h_1+\frac{1}{2}\, h_2+2\, h_3<h_5,\\
\varphi_{\Delta_1,h}(1,3)&=-\frac{5}{2}\, h_1+\frac{1}{2}\, h_2+3\, h_3<h_6.
\end{align*}
Then, we have the description
\begin{align*}
\mathcal{C}_{\Delta_1}=\{h=(h_1,\dots,h_6)\in\R^6 \, : \, \  h_1-h_2-h_3+h_4>0,\ 
&3h_1 - h_2 - 4h_3 + 2h_5>0,\\
&5h_1-h_2-6h_3+2h_6>0\}.
\end{align*}
In analogous way we can compute $\mathcal{C}_{\Delta_2}$:
\begin{align*}
\mathcal{C}_{\Delta_2}=\{h=(h_1,\dots,h_6)\in\R^6 \, : \, \  h_1-h_2-h_3+h_4>0,\ 
&2h_2 - h_3 - 3h_4 + 2h_5>0,\\
&4h_2-h_3-5h_4+2h_6>0\}.
\end{align*}
We observe that one of the inequalities appears twice, then $\mathcal{C}_{\Delta_1,\Delta_2}$ is defined by five inequalities. We can write:
\begin{align*}
\mathcal{C}_{\Delta_1,\Delta_2}=\{h=(h_1,\dots,h_6)\in\R^6 \, : \, \ \langle m_r, h \rangle  > 0 , \; r=1,\ldots,5\},
\end{align*}
where 
$m_1=(1,-1,-1,1,0,0)$, $m_2=(3,-1,-4,0,2,0)$, $m_3=(5,-1,-6,0,0,2)$, 
$m_4=(0,2,-1,-3,2,0)$,  $m_5=(0,4,-1,-5,0,2)$, and $\langle \, , \,  \rangle$ denotes the canonical inner product.
\end{ej}

\begin{remark} \rm Take any vector $h\in\R^n$ and consider the \textit{lower
convex hull} $\mathcal L$ of the $n$ lifted
points $(a_j,h_j)\in \R^{d+1}$, with $a_j$, $j=1, \dots, n$, in the support $\mathcal{A}$ (see Figure~\ref{fig:Regular triangulation}). Project to
 $\R^d$  the subsets of points  in each of the faces of $\mathcal L$. 
These subsets define a \textit{regular} 
subdivision  of $\mathcal{A}$ associated with $h$. When the lifting vector $h$ is generic, the regular
subdivision is a
 regular triangulation, in which all
 the subsets are simplices. 
Note that given a simplex $\Delta \subset \mathcal A$, the cone $\mathcal{C}_{\Delta}$ consists of the vectors $h$ for which
  the induced regular subdivision contains the simplex $\Delta$.
	
%There are supports for which not all the triangulations are regular (Figure \ref{fig:NonRegular triangulation}). 
%See \cite{triangulations} for a more detailed explanation of this topic.

\end{remark}

  \begin{center}
	\begin{figure}[h]
				\centering
			%\vfill
			\begin{tikzpicture}%
			[x={(1.595303cm, 0.021152cm)},
			y={(0.905733cm, 0.517961cm)},
			z={(-0.024093cm, 0.8770689cm)},
			scale=0.900000,
			back/.style={dotted,very thick},
			edge/.style={color=black, thick},
			facet/.style={fill=blue!95!black,fill opacity=0.800000},
			vertex/.style={inner sep=1.2pt,circle,draw=black,fill=black,thick,anchor=base}]
			%
			%
			%% Coordinate of the vertices:
			%%
			\coordinate (0.00000, 0.00000, 0.00000) at (0.00000, 0.00000, 0.00000);
			\coordinate (3.00000, 0.00000, 0.00000) at (3.00000, 0.00000, 0.00000);
			\coordinate (3.00000, 2.00000, 0.00000) at (3.00000, 2.00000, 0.00000);
			\coordinate (0.00000, 2.00000, 0.00000) at (0.00000, 2.00000, 0.00000);
			\coordinate (0.00000, 0.00000, 2.50000) at (0.00000, 0.00000, 2.50000);
			\coordinate (3.00000, 0.00000, 2.50000) at (3.00000, 0.00000, 2.50000);
			\coordinate (3.00000, 2.00000, 2.50000) at (3.00000, 2.00000, 2.50000);
			\coordinate (0.00000, 2.00000, 2.50000) at (0.00000, 2.00000, 2.50000);
			\coordinate (2.00000, 1.00000, 0.00000) at (2.00000, 1.00000, 0.00000);
			\coordinate (1.00000, 0.70000, 0.00000) at (1.00000, 0.70000, 0.00000);
			\coordinate (1.00000, 1.30000, 0.00000) at (1.00000, 1.30000, 0.00000);
			\coordinate (2.00000, 1.00000, 1.00000) at (2.00000, 1.00000, 1.00000);
			\coordinate (1.00000, 0.70000, 1.30000) at (1.00000, 0.70000, 1.30000);
			\coordinate (1.00000, 1.30000, 2.30000) at (1.00000, 1.30000, 2.30000);
			
			%% Drawing the facets
			%%
			%\fill[facet, fill=green!70!white,fill opacity=0.800000] (1.00000, 2.00000, 0.00000) -- (1.00000, 1.00000, 0.00000) -- (0.00000,
			%0.00000, 0.00000) -- cycle {};
			%\fill[facet, fill=white,fill opacity=0.800000] (1.00000, 2.00000, 0.00000) -- (0.00000, 1.00000, 0.00000) -- (0.00000,
			%0.00000, 0.00000) -- cycle {};
			%\fill[facet, fill=white,fill opacity=0.800000] (1.00000, 0.00000, 0.00000) -- (1.00000, 1.00000, 0.00000) -- (0.00000,
			%0.00000, 0.00000) -- cycle {};
			%\fill[facet, fill=red!70!white,fill opacity=0.800000] (0.00000, 1.00000, 2.00000) -- (1.00000, 2.00000, 0.00000) -- (0.00000,
			%0.00000, 0.00000) -- cycle {};
			%\draw[edge] (0.00000, 0.00000, 0.00000) -- (1.00000, 2.00000, 0.00000);
			%\fill[facet, fill=blue!70!white,fill opacity=0.800000] (1.00000, 0.00000, 1.00000) -- (1.00000, 1.00000, 0.00000) -- (0.00000,
			%0.00000, 0.00000) -- cycle {};
			
			%% Drawing edges in the front
			%%

			\draw[edge,back] (0.00000, 0.00000, 0.00000) -- (0.00000, 0.00000, 2.50000);
			\draw[edge,back](0.00000, 2.00000, 0.00000) -- (0.00000, 2.00000, 2.50000);
			\draw[edge,back](3.00000, 2.00000, 0.00000) -- (3.00000, 2.00000, 2.50000);
			\fill[facet, fill=blue!,fill opacity=1.000000] (0.00000, 2.00000, 2.50000) -- (1.00000, 0.70000, 1.30000) -- (2.00000, 1.00000, 1.00000) -- cycle {};
			\fill[facet, fill=green!,fill opacity=1.000000] (0.00000, 2.00000, 2.50000) -- (3.00000, 2.00000, 2.50000) -- (2.00000, 1.00000, 1.00000) -- cycle {};
			\fill[facet, fill=violet!80!white,fill opacity=1.000000] (0.00000, 2.00000, 2.50000) -- (1.00000, 0.70000, 1.30000) -- (0.00000, 0.00000, 2.50000) -- cycle {};
			\draw[edge] (0.00000, 2.00000, 2.50000) -- (2.00000, 1.00000, 1.00000);
			\fill[facet, fill=blue!50!white,fill opacity=1.000000] (0.00000, 2.00000, 0.00000) -- (1.00000, 0.70000, 0.00000) -- (2.00000, 1.00000, 0.00000) -- cycle {};
			\fill[facet, fill=green!50!white,fill opacity=1.000000] (0.00000, 2.00000, 0.00000) -- (3.00000, 2.00000, 0.00000) -- (2.00000, 1.00000, 0.00000) -- cycle {};
			\fill[facet, fill=violet!50!white,fill opacity=1.000000] (0.00000, 2.00000, 0.00000) -- (1.00000, 0.70000, 0.00000) -- (0.00000, 0.00000, 0.00000) -- cycle {};
			\draw[edge,back](1.00000, 1.30000, 0.00000) -- (1.00000, 1.30000, 2.30000);
			\draw[edge] (1.00000, 0.70000, 1.30000) -- (0.00000, 2.00000, 2.50000);
			\fill[facet, fill=yellow!,fill opacity=1.000000] (2.00000, 1.00000, 1.00000) -- (3.00000, 2.00000, 2.50000) -- (3.00000, 0.00000, 2.50000) -- cycle {};
			\fill[facet, fill=red!,fill opacity=1.000000] (0.00000, 0.00000, 2.50000) -- (1.00000, 0.70000, 1.30000) -- (3.00000, 0.00000, 2.50000) -- cycle {};
			\fill[facet, fill=orange!,fill opacity=1.000000] (3.00000, 0.00000, 2.50000) -- (1.00000, 0.70000, 1.30000) -- (2.00000, 1.00000, 1.00000) -- cycle {};

			\fill[facet, fill=yellow!50!white,fill opacity=1.000000] (2.00000, 1.00000, 0.00000) -- (3.00000, 2.00000, 0.00000) -- (3.00000, 0.00000, 0.00000) -- cycle {};
			\fill[facet, fill=red!50!white,fill opacity=1.000000] (0.00000, 0.00000, 0.00000) -- (1.00000, 0.70000, 0.00000) -- (3.00000, 0.00000, 0.00000) -- cycle {};
			\fill[facet, fill=orange!50!white,fill opacity=1.000000] (3.00000, 0.00000, 0.00000) -- (1.00000, 0.70000, 0.00000) -- (2.00000, 1.00000, 0.00000) -- cycle {};
			\draw[edge] (0.00000, 0.00000, 2.50000) -- (3.00000, 0.00000, 2.50000);
			\draw[edge] (0.00000, 0.00000, 2.50000) -- (1.00000, 0.70000, 1.30000);
			\draw[edge] (3.00000, 0.00000, 2.50000) -- (1.00000, 0.70000, 1.30000);
			\draw[edge] (2.00000, 1.00000, 1.00000) -- (1.00000, 0.70000, 1.30000);
			\draw[edge] (3.00000, 0.00000, 2.50000) -- (2.00000, 1.00000, 1.00000);
			\draw[edge] (0.00000, 0.00000, 0.00000) -- (3.00000, 0.00000, 0.00000);
			\draw[edge] (0.00000, 0.00000, 0.00000) -- (0.00000, 2.00000, 0.00000);
			\draw[edge] (3.00000, 0.00000, 0.00000) -- (3.00000, 2.00000, 0.00000);
			\draw[edge] (0.00000, 2.00000, 0.00000) -- (3.00000, 2.00000, 0.00000);
			\draw[edge] (3.00000, 0.00000, 0.00000) -- (2.00000, 1.00000, 0.00000);
			\draw[edge] (3.00000, 2.00000, 0.00000) -- (2.00000, 1.00000, 0.00000);
			\draw[edge] (0.00000, 2.00000, 0.00000) -- (2.00000, 1.00000, 0.00000);
			\draw[edge] (1.00000, 0.70000, 0.00000) -- (2.00000, 1.00000, 0.00000);
			\draw[edge] (1.00000, 0.70000, 0.00000) -- (3.00000, 0.00000, 0.00000);
			\draw[edge] (1.00000, 0.70000, 0.00000) -- (0.00000, 0.00000, 0.00000);
			\draw[edge] (1.00000, 0.70000, 0.00000) -- (0.00000, 2.00000, 0.00000);
			\draw[edge,back] (3.00000, 0.00000, 2.50000) -- (3.00000, 0.00000, 0.00000);
			\draw[edge,back](2.00000, 1.00000, 0.00000) -- (2.00000, 1.00000, 1.00000);
			\draw[edge,back] (1.00000, 0.70000, 0.00000) -- (1.00000, 0.70000, 1.30000);
			\draw[edge] (3.00000, 0.00000, 2.50000) -- (3.00000, 2.00000, 2.50000);
			\draw[edge] (3.00000, 2.00000, 2.50000) -- (2.00000, 1.00000, 1.00000);
			\draw[edge] (0.00000, 0.00000, 2.50000) -- (0.00000, 2.00000, 2.50000);
			\draw[edge] (3.00000, 2.00000, 2.50000) -- (0.00000, 2.00000, 2.50000);
			%% Drawing the vertices in the front
			%%
			\node[vertex] at (0.00000, 0.00000, 0.00000)     {};
			\node at (-0.2,0,0) {$a$};
			\node[vertex] at (0.00000, 2.00000, 0.00000)     {};
			\node[vertex] at (3.00000, 0.00000, 0.00000)     {};
			\node[vertex] at (3.00000, 2.00000, 0.00000)     {};
			\node at (4.05,0.6,0) {$\R^d$};
			\node at (4.2,0.6,2) {$\R^{d+1}$};
			\node[vertex] at (0.00000, 0.00000, 2.50000)     {};
			\node[vertex] at (0.00000, 2.00000, 2.50000)     {};
			\node[vertex] at (3.00000, 0.00000, 2.50000)     {};
			\node at (-0.6,0,2.5) {$(a,h(a))$};
			\node[vertex] at (3.00000, 2.00000, 2.50000)     {};
			\node[vertex] at (2.00000, 1.00000, 0.00000)   	 {};
			\node[vertex] at (1.00000, 1.30000, 0.00000) 	 {};
			\node[vertex] at (1.00000, 0.70000, 0.00000)	 {};
			\node[vertex] at (2.00000, 1.00000, 1.00000)   	 {};
			\node[vertex] at (1.00000, 0.70000, 1.30000) 	 {};
			\node[vertex] at (1.00000, 1.30000, 2.30000)	 {};

			\end{tikzpicture}
			% \vfill
			\caption{Regular triangulation.} \label{fig:Regular triangulation}
			\vspace{\baselineskip}
	\end{figure}
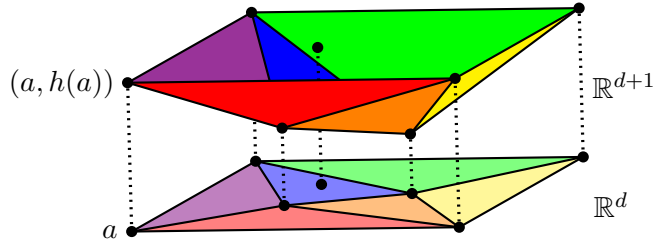
\end{center}

Given any $h\in\R^n$, consider the following family of polynomial systems parametrized
 by a positive real number $t$, which coincides for $t=1$ with
the system defined by the polynomials in~\eqref{systemf}:
\begin{equation}\label{systemwitht}
f_{1,t}(x)=\dots=f_{d,t}(x)=0,
\end{equation}
where 
\begin{equation*}
f_{i,t}(x)=\sum_{j=1}^n c_{ij}\, t^{h_j} \, x^{a_j} \in \R[x_1,\dots,x_d], \ i=1,\dots,d.
\end{equation*}
For each positive real value of $t$, this system has again support included in $\mathcal{A}$. Recall that
a common root of~\eqref{systemwitht} is nondegenerate when it is not a zero of the Jacobian of $f_{1,t}, \dots, f_{d,t}$.
The following result is a particular case of Theorem 3.4 in~\cite{bihan}.
 
 \begin{thm} \label{th:BS}
Let $\mathcal{A}=\{a_1, \dots, a_{n}\} \subset \Z^d$ be a finite point configuration and 
$C=(c_{ij})\in \R^{d\times n}$ a matrix. Let $\Delta_1, \Delta_2$ be two $d$-simplices with vertices 
in $\mathcal{A}$ which share a facet, and which are positively decorated by the matrix $C$. 
Let $h$ be any vector in the cone $\mathcal{C}_{\Delta_1,\Delta_2}$.
Then, there exists $t_0\in\R_{>0}$ such that for all $0<t<t_0$, the number of (nondegenerate) solutions of (\ref{systemwitht}) 
 contained in the positive orthant is at least two.
 \end{thm}

\begin{ej}\rm (Example \ref{example:support}, \ref{example:simplices} and \ref{example:cone} continued).
 \rm The simplices $\Delta_1$ and $\Delta_2$ are positively decorated by $C$ and share a facet. 
 Then, if we take $h\in\mathcal{C}_{\Delta_1,\Delta_2}$, there exists $t_0\in\R_{>0}$ such that for 
 all $0<t<t_0$, the number of (nondegenerate) solutions of the deformed system
\begin{eqnarray}\label{eq:systemexample21}
t^{h_1} - t^{h_2}2x^2 + t^{h_3}y + t^{h_4}x^2y - t^{h_5}xy^2 &=\ 0,\\
-t^{h_1}2 + t^{h_2}x^2 - t^{h_4}x^2y - t^{h_5}xy^2 + t^{h_6}xy^3 &=\ 0, \nonumber
\end{eqnarray}
is at least two.

Indeed, it is easy to check that the simplices 
$\Delta_4=\{(0,1),(1,2),(1,3)\}$ and $\Delta_5=\{(2,1),(1,2),(1,3)\}$ in Figure~\ref{fig:simplices} 
are also positively decorated by the matrix $C$. In fact,  by Theorem~2.9 in~\cite{AGB1},  for any
 $h\in \mathcal{C}_{\Delta_1}\cap \mathcal{C}_{\Delta_2}\cap \mathcal{C}_{\Delta_4}\cap
 \mathcal{C}_{\Delta_5}$ there exists $t_0$ such that if $0<t<t_0$ the system \eqref{eq:systemexample21}
  has at least four positive solutions. Here four is the number of simplices which are positively decorated. 
  For example if we take $h_1=1$, $h_2=h_3=h_4=0$, $h_5=1$ and $h_6=3$, we obtain the regular
   triangulation in Figure~\ref{fig:simplices}, and if we choose $t=1/12$, system (\ref{systemwitht})
    has four positive solutions. This fact can be checked using the free Computer
Algebra System Singular~\cite{Singular} 
    with the library ``signcond.lib" implemented by E. Tobis, with the following code:
\begin{verbatim}
> LIB "signcond.lib";
> ring r= 0, (x,y), dp;
> ideal i = 1/12-2*x^2+y+x^2*y-(1/12)*x*y^2, 
-2*(1/12)+x^2-x^2*y-(1/12)*x*y^2+(1/12)^3*x*y^3;
> ideal j = std(i);
> firstoct(j);
4
\end{verbatim}
Note that this procedure is symbolic and thus certified, as opposed to numeric
algorithms to compute the roots which can be affected by numerical unstability. 
It is based on the algorithms described in~\cite{BPR}.

\end{ej}

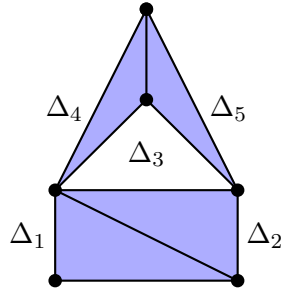
\begin{figure}
	\centering
	\begin{tikzpicture}
	[scale=1.200000,
	back/.style={loosely dotted, thin},
	edge/.style={color=black, thick},
	facet/.style={fill=red!70!white,fill opacity=0.800000},
	vertex/.style={inner sep=1.5pt,circle,draw=black,fill=black,thick,anchor=base
	}]
	%
	%
	%% Coordinate of the vertices:
	%%
	\coordinate (0.00000, 0.00000) at (0.00000, 0.00000);
	\coordinate (0.00000, 1.00000) at (0.00000, 1.00000);
	\coordinate (2.00000, 0.00000) at (2.00000, 0.00000);
	\coordinate (2.00000, 1.00000) at (2.00000, 1.00000);
	\coordinate (1.00000, 2.00000) at (1.00000, 2.00000);
	\coordinate (1.00000, 3.00000) at (1.00000, 3.00000);
	%%
	%%
	%% Drawing the interior
	%%
	\fill[facet, fill=blue!40!white,fill opacity=0.800000] (0.00000, 0.00000) --  (2.00000, 0.00000)-- (0.00000, 1.00000) -- cycle {};
	\fill[facet, fill=blue!40!white,fill opacity=0.800000] (2.00000, 0.00000)-- (0.00000, 1.00000) -- (2.00000, 1.00000) -- cycle {};
	\fill[facet, fill=blue!40!white,fill opacity=0.800000] (1.00000, 2.00000)-- (0.00000, 1.00000) -- (1.00000, 3.00000) -- cycle {};
	\fill[facet, fill=blue!40!white,fill opacity=0.800000] (1.00000, 2.00000)-- (2.00000, 1.00000) -- (1.00000, 3.00000) -- cycle {};
	
	%%
	%%
	%% Drawing edges
	%%
	\draw[edge] (0.00000, 0.00000) -- (0.00000, 1.00000);
	\draw[edge] (2.00000, 0.00000) -- (0.00000, 1.00000);
	\draw[edge] (2.00000, 1.00000) -- (0.00000, 1.00000);
	\draw[edge] (2.00000, 1.00000) -- (1.00000, 2.00000);
	\draw[edge] (1.00000, 3.00000) -- (1.00000, 2.00000);
	\draw[edge] (0.00000, 1.00000) -- (1.00000, 2.00000);
	\draw[edge] (0.00000, 0.00000) -- (2.00000, 0.00000);
	\draw[edge] (2.00000, 0.00000) -- (2.00000, 1.00000);
	\draw[edge] (2.00000, 1.00000) -- (1.00000, 3.00000);
	\draw[edge] (1.00000,3.00000) -- (0.00000, 1.00000);
	%%
	%%
	%% Drawing the vertices
	%%
	\node[vertex] at (0.00000, 0.00000){};
	\node[vertex] at (0.00000, 1.00000){};
	\node[vertex] at (2.00000, 0.00000){};
	\node[vertex] at (2.00000, 1.00000){};
	\node[vertex] at (1.00000, 2.00000){};
	\node[vertex] at (1.00000, 3.00000){};
	\node[vertex] at (1.00000, 3.00000){};
	\node at (-0.30000, 0.50000){$\Delta_1$};
	\node at (2.30000, 0.50000){$\Delta_2$};
    \node at (1.00000, 1.40000){$\Delta_3$};
	\node at (0.10000, 1.90000){$\Delta_4$};
	\node at (1.90000, 1.90000){$\Delta_5$};
	\end{tikzpicture}
	\caption{Simplices of $\mathcal{A}$ in Example~\ref{example:support}, which are positively decorated by the matrix $C$.}
	\label{fig:simplices}
\end{figure}

 %\begin{remark}
 %If we adopt the convention that $h$ is concave in the definition of a regular polyhedral subdivision, 
 %then one has to take $t>1/t_0$ in Theorem \ref{th:BS} (since $t^{h(a_j)}=(1/t)^{-h(a_j)}$).
 %\end{remark}
  We now state a similar result, but here we describe a subset  with nonempty interior in the space of coefficients
 where we can find at least two positive solutions of the associated system. This is
 a simplified version of Theorem~2.11 in~\cite{AGB1}.

 \begin{thm}\label{thm:main2}
 Consider a set $\mathcal{A}=\{a_1, \dots, a_{n}\}$ of $n$ points in $\R^d$ and 
 a matrix $C=(c_{i,j}) \in \R^{d \times n}$.  Assume  there are two $d$-simplices $\Delta_1$, $\Delta_2$ with vertices in
 $\mathcal{A}$, which share a facet
and are positively decorated by the matrix $C$.
 Assume that the cone ${\mathcal C}_{\Delta_1,\Delta_2}$ is defined by the inequalities
 \begin{equation} \label{E:for2}
 \langle m_r, h \rangle  > 0 , \; r=1,\ldots,\ell,
 \end{equation}
 where $m_r=(m_{r,1},\ldots,m_{r,n}) \in \R^n$.

Then, there exists constants $M_1,\dots,M_{\ell}>0$ such that for any $\gamma$ in the open set
\[ U \, = \{ \gamma \in \R_{>0}^n \, : \, \gamma^{m_r} < M_r, \, r=1 \dots,\ell\},\]
 the system 
 \begin{equation}\label{systemgamma}
 \sum_{j=1}^n c_{ij}\gamma_{j}x^{a_j}=0 , \;  \ i=1,\dots,d,
 \end{equation}
 has at least $2$ nondegenerate solutions in the positive orthant.
  \end{thm}

\begin{remark} \rm
Note that the choice of the positive constants $M_1, \dots, M_{\ell}$ is {\it not} algorithmic,
 but we describe an open set in coefficient space for which more than 
 one positive solution can be found. Furthermore,  inequalities~\eqref{E:for2} 
 indicate how to scale the coefficients of the system  in order to get  at least two positive solutions.
\end{remark}

 \section{Enzymatic cascades with two layers}\label{sec:2}
    
In this section we work with the case of an enzymatic cascade with two layers, 
and then in Section~\ref{sec:3} we will work with the general case.
The network involves two phosphorylation cycles.  We call $S_1$ and $S_2$
 the substrate proteins in the first and second layers respectively. The upper index can be interpreted 
 as the absence ($0$) or the presence ($1$) of a phosphate group. 
 The phosphorylation in the first layer is catalyzed by the enzyme $E$. The activated protein $S^1_1$ 
 in the first layer acts as the modifier enzyme in the second layer, which is depicted in (A) in Figure~\ref{fig:sadi}.
 
%  \begin{figure}[h]
% \centering
% \begin{tikzpicture}[scale=0.7,node distance=0.7cm] 
%  \node[] at (1.1,-1.5) (dummy2) {};
%  \node[left=of dummy2] (p0) {$S^0_2$};
%  \node[right=of p0] (p1) {$S^1_2$}
%    edge[->, bend left=45] node[below] {\textcolor{black!70}{\small{$F$}}} (p0)
%    edge[<-, bend right=45] node[above] {} (p0); 
%  \node[] at (-1,0) (dummy) {};
%  \node[left=of dummy] (s0) {$S^0_1$};
%  \node[right=of s0] (s1) {\textcolor{red}{$S^1_1$}}
%    edge[->, bend left=45] node[below] {\textcolor{black!70}{\small{$F$}}} (s0)
%    edge[<-, bend right=45] node[above] {\textcolor{black!70}{\small{$E$}}} (s0)
%    edge[->,thick,color=red, bend left=25] node[above] {} ($(p0.north)+(25pt,10pt)$); 
% \end{tikzpicture}
%   \caption{Enzymatic cascade with $2$ layers.}
% \label{fig:cascade2}
% \end{figure}
 Note that the dephosphorylation is carried out by the same phosphatase $F$, 
 which as  we pointed out in the Introduction 
 gives the capacity for multistationarity to the network by~\cite{feliu}. 
 The kinetics of this network is deduced by applying the law of mass-action to the following labeled digraph:  
 \begin{eqnarray} \label{eq:cascade2}
 %\nonumber 
 S^0_1 + E 
 \arrowschem{k_{\rm{on}_1}}{k_{\rm{off}_1}} Y^0_1
 \stackrel{k_{\rm{cat}_1}}{\rightarrow} S^1_1+E & \ \ &
 S^0_2 + S^1_1 \arrowschem{k_{\rm{on}_{2}}}{k_{\rm{off}_{2}}} Y^0_2
 \stackrel{k_{\rm{cat}_{2}}}{\rightarrow} S^1_2 + S^1_1 \\
 %\label{eq:cascade_2_complete} \\ 
 \nonumber 
 S^1_1 + F \arrowschem{\ell_{\rm{on}_{1}}}{\ell_{\rm{off}_{1}}} Y^1_1
 \stackrel{\ell_{\rm{cat}_{1}}}{\rightarrow} S^0_1+ F
 & \ \  &
 S^1_2 + F \arrowschem{\ell_{\rm{on}_{2}}}{\ell_{\rm{off}_{2}}} Y^1_2
 \stackrel{\ell_{\rm{cat}_{2}}}{\rightarrow} S^0_2+ F.
 \end{eqnarray}

 We denote by $Y^0_1$, $Y^0_2$, $Y^1_1$, $Y^1_2$ the intermediate complexes, which consist of a 
 single chemical species formed by the union of the substrate with the enzyme. 
 The concentrations of the species will be denoted with small letters, for example $s^0_1$ will denote the concentration of $S^0_1$.
% We denote $s^0_1$, $s^1_1$, $s^0_2$, $s^1_2$,  $e$, $f$, $y^0_1$, $y^1_1$, $y^0_2$, $y^1_2$ the concentration of 
% the ten chemical species in the network $S^0_1$, $S^1_1$, $S^0_2$, $S^1_2$, $E$, $F$, $Y^0_1$, $Y^1_1$, $Y^0_2$, $Y^1_2$, 
% respectively.
 The associated dynamical system that arises under mass-action kinetics equals:
 \vskip -10pt
 \begin{align*}
 \frac{ds^0_1}{dt} = & {-k_{\rm{on}_1}}s^0_1e + {k_{\rm{off}_1}}y^0_1 + {\ell_{\rm{cat}_1}} y^1_1, & \frac{dy^1_1}{dt} = & 
 {\ell_{\rm{on}_1}}s^1_1 f -{(\ell_{\rm{off}_1}+\ell_{\rm{cat}_1})}y^1_1, \\
 \frac{ds^1_1}{dt} = & {k_{\rm{cat}_1}}y^0_1 - {\ell_{\rm{on}_{1}}}s^1_1f + {\ell_{\rm{off}_{1}}}y^1_1    & \frac{dy^0_2}{dt} = 
 &{k_{\rm{on}_2}}s^0_2s^1_1  -{(k_{\rm{off}_2}+k_{\rm{cat}_2})}y^0_2,\\
  +&{k_{\rm{on}_{2}}}s^0_2s^1_1 + ({k_{\rm{off}_{2}}} +
   {k_{\rm{cat}_2}})y^0_2, & \frac{dy^1_2}{dt} = & {\ell_{\rm{on}_2}}s^1_2 f 
 -{(\ell_{\rm{off}_2}+\ell_{\rm{cat}_2})}y^1_2, \\
 \frac{ds^0_2}{dt} = & {-k_{\rm{on}_2}}s^0_2s^1_1 + {k_{\rm{off}_2}}y^0_2 +
  {\ell_{\rm{cat}_2}} y^1_2, &  \frac{de}{dt} = 
 &-{k_{\rm{on}_1}}s^0_1 e + {(k_{\rm{off}_1}+ k_{\rm{cat}_1})}y^0_1, \\
 \frac{ds^1_2}{dt} = & {k_{\rm{cat}_2}}y^0_2 - {\ell_{\rm{on}_{2}}}s^1_2f + 
 {\ell_{\rm{off}_{2}}}y^1_2, & \frac{df}{dt} = & - 
 {\ell_{\rm{on}_{1}}}s^1_1f   + {(\ell_{\rm{off}_1}+\ell_{\rm{cat}_1})}y^1_1\\
 \frac{dy^0_1}{dt} = & {k_{\rm{on}_1}}s^0_1 e -{(k_{\rm{off}_1}+
 k_{\rm{cat}_1})}y^0_1, & & -{\ell_{\rm{on}_2}}s^1_2 
 f+{(\ell_{\rm{off}_2}+\ell_{\rm{cat}_2})}y^1_2.\\
 \end{align*}
 
 In this case, there is a basis of the conservation laws given by the four linear equations:
 \begin{align}\label{eq:conscascada2}
 e + y^0_1 = & E_{tot}, \nonumber \\
 f + y^1_1 + y^1_2 = & F_{tot},\\
 s^0_1 + s^1_1 + y^0_1 + y^1_1 + y^0_2 = & S_{1,tot}, \nonumber\\
 s^0_2 + s^1_2 + y^0_2 + y^1_2 = & S_{2,tot}. \nonumber
 \end{align}

Enzymatic cascades are an example of \emph{$s$-toric MESSI networks},
defined in \cite{aliciaMer}. By
Theorem 3.5 in \cite{aliciaMer} we can find binomial equations that describe the steady
states. This is a general procedure, that in this case is easily obtained by manipulating 
 the differential equations. First, the concentrations of the intermediates 
 species $y^0_1, y^1_1, y^0_2, y^1_2$ at steady state satisfy the following binomial equations:
 %\vskip -10pt
 \begin{eqnarray}\label{eq:intermedios}
  y^0_1 \, - \, K_1 \, es^0_1 =0, & & y^1_1\, - \, L_1 \, fs^1_1 =0,\\
 y^0_2\, - \, K_2 \, s^1_1s^0_2 =0, & &  y^1_2\, - \,  L_2 \, fs^1_2 =0,\nonumber
 \end{eqnarray}
where $K_1=\frac{k_{\rm{on}_1}}{k_{\rm{off}_1}+k_{\rm{cat}_1}}$, 
$K_2=\frac{k_{\rm{on}_2}}{k_{\rm{off}_2}+k_{\rm{cat}_2}}$, $L_1=
\frac{\ell_{\rm{on}_1}}{\ell_{\rm{off}_1}+\ell_{\rm{cat}_1}}$ and  
$L_2=\frac{\ell_{\rm{on}_2}}{\ell_{\rm{off}_2}+\ell_{\rm{cat}_2}}$ ($K_1^{-1}$, 
$K_2^{-1}$, $L_1^{-1}$ and $L_2^{-1}$ are usually called \textit{Michaelis-Menten constants}).
The whole steady state variety can be cut out in the positive orthant by adding to
the binomials in \eqref{eq:intermedios}, the following binomial equations:  
 \[\tau_1\, s^0_1\, e - \nu_1\, s^1_1\, f =0, \quad \tau_2\, s^0_2\, s^1_1 - \nu_2\, s^1_2\, f =0,\]
 where $\tau_1 = k_{\rm{cat}_1}\, K_1$, $\tau_2=k_{\rm{cat}_2}\, K_2$, 
 $\nu_1=\ell_{\rm{cat}_1}\, L_1$ and $\nu_2=\ell_{\rm{cat}_2}\, L_2$. 
 
 Therefore, we can parametrize the positive steady states by monomials. For instance, 
 we can write the concentration at steady state of $s^0_1, s^0_2$  and the intermediate species, 
 in terms of the species $(e, f, s^1_1, s^1_2)$:
 \begin{align}\label{parametrizationcascade2}
&s^0_1=G_1\, \frac{s^1_1\, f}{e}, \qquad y^0_1=K_1\, G_1\, s^1_1\, f, \qquad y^1_1=L_1\, s^1_1\, f,\\
&s^0_2=G_2\, \frac{s^1_2\, f}{s^1_1}, \qquad y^0_2=K_2\, G_2\, s^1_2\, f, \qquad y^1_2=L_2\, s^1_2\, f, \nonumber
 \end{align}
where $G_1=\frac{\nu_1}{\tau_1}$ and $G_2=\frac{\nu_2}{\tau_2}$.

Now, we apply our results to this case.  Denote by
\begin{equation}\label{eq:A12}
A_1 =\frac{\ell_{\rm{cat}_1}}{k_{\rm{cat}_1}}, \quad A_2=  \frac{\ell_{\rm{cat}_2}}{k_{\rm{cat}_2}},
\end{equation}
and assume that $S_{1,tot}, S_{2,tot}, E_{tot}, F_{tot}>0$.
Consider the following rational functions $\alpha_1, \alpha_2, \alpha_3, \alpha_4$ 
depending on the  catalytic reaction rate constants and total concentration constants:
\begin{equation*}
\begin{aligned}
\alpha_1=&\dfrac{S_{1,tot}}{F_{tot}}- A_2,\\
\alpha_2=&(A_1+1) - \dfrac{S_{1,tot}}{F_{tot}},\\
\alpha_3=&\dfrac{A_1 +1-A_2}{A_1}\, \dfrac{E_{tot}}{F_{tot}} - 
 \left(\dfrac{S_{1,tot}}{F_{tot}} - A_2\right),\\
\alpha_4=&
\dfrac{A_1+1-A_2}{A_2+1}\, \dfrac{S_{2,tot}}{F_{tot}} - 
 \left(A_1+1-\dfrac{S_{1,tot}}{F_{tot}}\right).
\end{aligned}
\end{equation*}

We then have:

\begin{thm}\label{th:cascade2}
	Consider the enzymatic cascade with two layers with digraph as in \eqref{eq:cascade2} 
	and let $A_1, A_2$ as in~\eqref{eq:A12}.
	Assume that the reaction rate constants verify $A_1+1>A_2$ and the total 
	concentration constants verify the  inequalities
	$ \alpha_1, \alpha_2, \alpha_3, 
	\alpha_4>0,$
	that is:
	
	$$A_1+1 > \dfrac{S_{1,tot}}{F_{tot}} > A_2, \quad 
	\dfrac{E_{tot}}{F_{tot}} > \left(\dfrac{S_{1,tot}}{F_{tot}}- A_2\right) \, \dfrac{A_1}{A_1+1-A_2},$$ %\quad 
$$\dfrac{S_{2,tot}}{F_{tot}} > \left(A_1+1- \dfrac{S_{1,tot}}{F_{tot}}\right) \, \dfrac{A_2+1}{A_1+1-A_2},$$
or instead, that $A_1+1< A_2$ and $ \alpha_1, \alpha_2, \alpha_3, \alpha_4<0.$

	Fix generic positive numbers $h_2$, $h_3$, $h_7$, $h_8$ such that $h_8 < h_2$.
	Then, there exists $t_0 >0$ such that for any value of $t \in (0, t_0)$ the system has at least two positive steady states after 
	modifying the coefficients $k_{\rm{on}_1}, k_{\rm{on}_2}, \ell_{\rm{on}_1}, \ell_{\rm{on}_2}$ via   the rescaling 
	$t^{-h_7}k_{\rm{on}_1}$, $t^{-h_3-h_8}k_{\rm{on}_2}$, $t^{-h_2-h_3}\ell_{\rm{on}_1}$ and 
	$t^{-h_2}\ell_{\rm{on}_2}$. 
	
	Also, for any fixed choice of reaction rate constants and total concentration constants
	 lying in the open set defined by one of the previous set of inequalities, 
	there exist positive constants $M_1, \dots, M_6$ such that for any values of $\beta_1, \beta_2, \eta_1, \eta_2$ verifying 
\begin{equation}\label{inequalitiesbeta}
\frac{1}{\eta_2} < M_1,\quad \frac{\eta_2}{\eta_1}<M_2,\quad \frac{1}{\beta_1}<M_3, 
\quad \frac{\eta_1}{\eta_2\beta_2}<M_4,\quad \frac{\beta_2}{\eta_1}<M_5,\quad \frac{1}{\beta_2}<M_6,
\end{equation}
the rescaling of the given parameters $k_{\rm{on}_0}$, $k_{\rm{on}_1}$, $\ell_{\rm{on}_0}$ 
and $\ell_{\rm{on}_1}$ by $\beta_1k_{\rm{on}_1}$, $\beta_2k_{\rm{on}_2}$, $\eta_1\ell_{\rm{on}_1}$ 
and $\eta_2\ell_{\rm{on}_2}$ respectively, gives raise to a multistationary system.
\end{thm}

\begin{proof}
		We substitute the monomial parametrization in (\ref{parametrizationcascade2})  
	 of the steady states in terms of the concentrations of the species
	 $e, f, s^1_1, s^1_2$    into the linear
	 conservation relations \eqref{eq:conscascada2}. We write this system in matricial form:
	\[C\begin{pmatrix}
	e & f & s^1_1 & s^1_2 & s^1_1f & s^1_2f & s^1_1fe^{-1} & s^1_2f(s^1_1)^{-1} & 1
	\end{pmatrix}^t=0,
	\]
	where the matrix $C\in \R^{4\times 9}$  of coefficients equals:
	\begin{equation}\label{matrixcascade2}
	C=\begin{pmatrix}
	1 & 0 & 0 & 0 & K_1G_1 & 0 & 0 & 0 & -E_{tot}\\
	0 & 1 & 0 & 0 & L_1 & L_2 & 0 & 0 & -F_{tot}\\
	0 & 0 & 1 & 0 & K_1G_1 + L_1 & K_2G_2 & G_1 & 0 & -S_{1,tot} \\
	0 & 0 & 0 & 1 & 0 & K_2G_2 + L_2 & 0 & G_2 & -S_{2,tot}
	\end{pmatrix}.
	\end{equation}
	
	If we order the variables as before, the support of this system is:
	\begin{align*}
	\mathcal{A}=\{(1,0,0,0), (0,1,0,0), (0,0,1,0), (0,0,0,1),  (0,1,1,0),  \\ (0,1,0,1), (-1,1,1,0), (0,1,-1,1), (0,0,0,0) \}.
	\end{align*}
	
	We want to find two positively decorated $4$-simplices with vertices in $\mathcal{A}$ which share one facet. For example 
we take the 
	simplices 
	\begin{align*}
	\Delta_1&=\{(1,0,0,0),(0,0,0,1), (0,1,1,0), (0,1,0,1), (0,0,0,0) \},\\
	\Delta_2&=\{(1,0,0,0), (0,1,1,0), (0,1,0,1), (0,1,-1,1), (0,0,0,0) \}.
	\end{align*}
	
	It is straightforward to check that both simplices are positively decorated by $C$ 
	if either $A_1 + 1 > A_2$ 
	and $\alpha_1, \alpha_2, \alpha_3, \alpha_4>0$, or 
	$A_1 + 1 < A_2$ and
	$\alpha_1, \alpha_2, \alpha_3, \alpha_4<0$, as in the statement.

Given $h\in\mathcal{C}_{\Delta_1,\Delta_2}$, by Theorem~\ref{th:BS}, there exists $t_0 \in \R_+$ such that for all $0<t<t_0$, 
	the number of positive (nondegenerate) solutions of the scaled system:
	\vskip -10pt
	\begin{small}
		\begin{equation}\label{systemwitht-cascade2}
		\begin{aligned}
		t^{h_1} \, e + t^{h_5} \, K_1G_1s^1_1f -  t^{h_9} \, E_{tot}  &=& 0, \\
		t^{h_2}\, f + t^{h_5} \, L_1 f s^1_1 + t^{h_6}  \, L_2 f s^1_2 -  t^{h_9} \, F_{tot} &=& 0,\\
		t^{h_3} \, s^1_1 + t^{h_7} \, G_1 \dfrac{s^1_1 f}{e} + t^{h_5} \, 
		 (K_1G_1 + L_1) s^1_1 f  + t^{h_6} \, K_2 G_2 s^1_2 f - t^{h_9} 
		\, S_{1,tot} &=& 0, \\
		t^{h_4} \, s^1_2 + t^{h_8} \, G_2 \dfrac{s^1_2f}{s^1_1} + t^{h_6} \, (K_2 G_2 + L_2)s^1_2f - t^{h_9} \, S_{2,tot}&=& 0, \\
		\end{aligned}
		\end{equation} \end{small}
	
	\noindent is at least two. 
	If we think of the vector $h$ as a function $\mathcal A \rightarrow \R$ (defined by $h(a_j)=h_j$),
	then
	$h_1=h(1,0,0,0)$, $h_2=h(0,1,0,0)$, $h_3=h(0,0,1,0)$, $h_4=h(0,0,0,1)$,  $h_5=h(0,1,1,0)$, 
	$h_6=h(0,1,0,1)$, $h_7=h(-1,1,1,0)$, $h_8=h(0,1,-1,1)$
	and $h_9=h(0,0,0,0)$.
	Let $\varphi_1$ and $\varphi_2$ be the affine linear 
	functions which agree with $h$ on the simplices $\Delta_1$ and 
	$\Delta_2$ respectively.
	%($\varphi_i=\varphi_{\Delta_i,h}$ in our previous notations).
	We can take $h_1=h_4=h_5=h_6=h_9=0$.
		Then $\varphi_1=0$, $h_8>0$ and $\varphi_2$ is defined by
	$\varphi_2(x,y,z,w)=h_8y - h_8z - h_8w$.
	Moreover,
	
	\vskip -10pt
	\[\begin{matrix}
	0 < h_2,&  \varphi_2(0,1,0,0)  & = & h_8 & < & h_2,\\
	0 < h_3,&  \varphi_2(0,0,1,0)  & = & -h_8 & < & h_3,\\
	0 < h_7,&  \varphi_2(-1,1,1,0)  & = & 0 & < & h_7, 
	\end{matrix}\]
	where we could take $h_2, h_3$ and $h_7$ generic.
	
	If we change the variables $\bar f = t^{h_2}f$, $\bar s^1_1 = t^{h_3}s^1_1$, 
	we get the following (Laurent) polynomial equations:
	\begin{equation}\label{systemwitht-cascade2-c}
	\begin{aligned}
	e + t^{-h_2-h_3}\, K_1G_1\, \bar s^1_1\bar f -  E_{tot}  &=& 0, \\
	\bar f + t^{-h_2-h_3}\, L_1\, \bar f\bar s^1_1 + t^{-h_2\, }L_2\, \bar f s^1_2 - F_{tot} &=& 0,\\
	\bar s^1_1 + t^{h_7-h_2-h_3}\, G_1\, \dfrac{\bar s^1_1\bar f}{e} +
	 t^{-h_2-h_3}\, (K_1G_1 + L_1)\, \bar s^1_1\bar f  + t^{-h_2}\, K_2 G_2 \, 
	s^1_2\bar f - S_{1,tot} &=& 0, \\
	s^1_2 + t^{h_8+h_3-h_2}\, G_2\dfrac{s^1_2\bar f}{\bar s^1_1} +
	 t^{-h_2}\, (K_2 G_2 + L_2)\,  s^1_2\bar f - S_{2,tot}&=& 0.
	\end{aligned}
	\end{equation}
	It is straightforward to verify that if we scale the constants:
	\begin{equation}\label{eq:const}
t^{-h_7}K_1,\ 
	t^{-h_3-h_8}K_2,\
	t^{-h_2-h_3}L_1,\
	t^{-h_2}L_2,
	\end{equation}
	and we keep fixed the values of  $k_{\rm{cat}_1}$, $k_{\rm{cat}_2}$, 
	${\ell_{\rm{cat}_1}}$ and ${\ell_{\rm{cat}_2}}$ and the total 
	values $E_{tot}$, $F_{tot}$, $S_{1,tot}$ and $S_{2,tot}$,
	 the intersection of the steady state variety and the linear varieties 
	of fixed total concentrations of the dynamical system associated 
	with the corresponding network, is described by system 
	(\ref{systemwitht-cascade2-c}). 
	
	It is easy to
	check that to get the scaling in~\eqref{eq:const}, it is enough to rescale the original constants as follows:
	$t^{-h_7}k_{\rm{on}_1}$, $t^{-h_3-h_8}k_{\rm{on}_2}$,
	 $t^{-h_2-h_3}\ell_{\rm{on}_1}$ and $t^{-h_2}\ell_{\rm{on}_2}$. 
	Then, for these choices of constants the system has at least two positive steady states. 
	The last part of the statement follows from the previous rescaling or from the inequalities 
	that define the cone $\mathcal{C}_{\Delta_1,\Delta_2}$ of heights inducing regular subdivisions 
	of the convex hull of $\mathcal{A}$ that contain $\Delta_1$ and $\Delta_2$ and Theorem~\ref{thm:main2}. 
	For instance, we can check that $\mathcal{C}_{\Delta_1,\Delta_2}$ is defined by $6$ inequalities. We can write:
\begin{align*}
\mathcal{C}_{\Delta_1,\Delta_2}=\{h=(h_1,\dots,h_8)\in\R^8 \, : \,  \langle m_r, h \rangle  > 0 , \; r=1,\ldots,6\},
\end{align*}
where $\langle ,  \rangle$ denotes the canonical inner product of $\R^8$ and $m_1=(0,1,0,1,0,-1,0,0,-1)$, 
$m_2=(0,0,1,-1,-1,1,0,0,0),  m_3=(1,0,0,0,-1,0,1,0,-1), m_4=(0,0,0,1,1,-2,0,1,-1)$, $m_5=(0,1,0,0,-1,1,0,-1,0),
 m_6=(0,0,1,0,0,-1,0,1,-1)$. By Theorem~\ref{thm:main2}, there exist $M_1,\dots,M_6>0$ such that for any 
 $\gamma=(\gamma_1,\dots,\gamma_9)$ in the open set \[ U \, =
  \{ \gamma \in \R_{>0}^9 \, : \, \gamma^{m_r} < M_r, \, r=1 \dots,6\},\] the system 
	\vskip -10pt
	\begin{small}
		\begin{equation}\label{systemwithgamma-cascade2}
		\begin{aligned}
		\gamma_1 \, e + \gamma_5 \, K_1G_1\, s^1_1f -  \gamma_9 \, E_{tot}  &=& 0, \\
		\gamma_2\, f + \gamma_5\, L_1\,  f s^1_1 + \gamma_6  \, L_2\,  f s^1_2 -  \gamma_9 \, F_{tot} &=& 0,\\
		\gamma_3 \, s^1_1 + \gamma_7 \, G_1 \dfrac{s^1_1 f}{e} +
		 \gamma_5 \,  (K_1G_1 + L_1) \, s^1_1 f  + \gamma_6 \, K_2 G_2 \, s^1_2 f - \gamma_9 
		\, S_{1,tot} &=& 0, \\
		\gamma_4 \, s^1_2 + \gamma_8 \, G_2 \dfrac{s^1_2f}{s^1_1} +
		 \gamma_6 \, (K_2 G_2 + L_2)\, s^1_2f - \gamma_9 \, S_{2,tot}&=& 0, \\
		\end{aligned}
		\end{equation} \end{small}
	
	\noindent has at least two positive solutions. If we take $\gamma_1=\gamma_2=\gamma_3=
	\gamma_4=\gamma_9=1$, and we denote $\beta_1=\frac{\gamma_5}{\gamma_7}$, $\beta_2=
	\frac{\gamma_6}{\gamma_8}$, $\eta_1=\gamma_5$ and $\eta_2=\gamma_6$, the conditions such that  
	$\gamma$ belongs to $U$ are equivalent to the conditions \eqref{inequalitiesbeta}, and it is easy to check that the 
	steady state equations of the network after the rescaling of the given parameters $k_{\rm{on}_0}$, 
	$k_{\rm{on}_1}$, $\ell_{\rm{on}_0}$ and $\ell_{\rm{on}_1}$ by $\beta_1\, k_{\rm{on}_1}$, 
	$\beta_2\, k_{\rm{on}_2}$, $\eta_1\, \ell_{\rm{on}_1}$ and $\eta_2\, 
	\ell_{\rm{on}_2}$ give system \eqref{systemwithgamma-cascade2}.
	
\end{proof}

%\begin{remark}\label{remark:constantscompatible} \rm
%	Note that the inequalities in the statement of Theorem~\ref{th:cascade2} are clearly compatible. 
%	For example, they are satisfied if we take in the first case 
%	$\frac{\ell_{\rm{cat}_1}}{k_{\rm{cat}_1}}=1$, $\frac{\ell_{\rm{cat}_2}}{k_{\rm{cat}_2}}=1$, $E_{tot}=F_{tot}=20$, 
%	$S_{1,tot}=S_{2,tot}=30$, 
%	and in the second case  $\frac{\ell_{\rm{cat}_1}}{k_{\rm{cat}_1}}=1$, $\frac{\ell_{\rm{cat}_2}}{k_{\rm{cat}_2}}=3$, 
%	$E_{tot}=F_{tot}=20$, $S_{1,tot}=S_{2,tot}=50$. 
%\end{remark}

\begin{ej} \rm 	Note that the inequalities in the statement of Theorem~\ref{th:cascade2} are clearly compatible. 
For example, the inequalities are satisfied if we take in the first case 
	$\frac{\ell_{\rm{cat}_1}}{k_{\rm{cat}_1}}=1$, $\frac{\ell_{\rm{cat}_2}}{k_{\rm{cat}_2}}=1$, $E_{tot}=F_{tot}=20$, 
	$S_{1,tot}=S_{2,tot}=30$. We can obtain in this case
a value of $t$ such that the system \eqref{systemwitht-cascade2-c} has two or more positive solutions, 
using  Singular. 
Fix for example, $h_2=2$, $h_3=1$, $h_7=1$, $h_8=1$, $K_1=1$, $K_2=1$, $L_1=1$ and $L_2=1$.
We have then that $G_1=1$ and $G_2=1$. If we take 
$t=\frac{1}{24}$, we have that the system has 3 positive solutions:
\begin{verbatim}
>LIB "signcond.lib";
>ring r=(0,t), (x,y,z,w), dp;
>poly f1=x+t^3*y*z-20;
>poly f2=y+t^3*y*z+t^2*y*w-20;
>poly f3=x*z+t^2*z*y+t^3*2*y*z*x+t^2*y*w*x-30*x;
>poly f4=z*w+y*w+t^2*y*z*w-30*z;
>poly g1=subst(f1,t,24);
>poly g2=subst(f2,t,24);
>poly g3=subst(f3,t,24);
>poly g4=subst(f4,t,24);
>ideal i=g1,g2,g3,g4;
>ideal j=std(i);
>firstoct(j);
3
\end{verbatim}
Here $x=e$, $y=\bar{f}$, $z={\bar s^1_1}$ and $w=s^1_2$. 
It can be checked that if we take a slighly higher value $t=\frac{1}{23}$, 
the corresponding system has only one positive solution.
\end{ej}

 \section{Enzymatic cascades with $n$ layers}\label{sec:3}

We now present our results to the general case of an enzymatic cascade of $n$ layers, where we have 
 $n$ phosphorylation cycles (as in Figure \ref{fig:cascaden}), under
 the assumption that there are (at least) two layers which share a phosphatase. 
  We separate our study into two cases: 
 the case of the occurrence of the same phosphatase in two consecutive layers  (see Theorem~\ref{th:provisorio-cascaden})
and the case  where the layers which share the phosphatase are not consecutive
 (see Theorem~\ref{th:cascadan-noconsecutivos}).
  As we pointed out in the Introduction, the difficulty to deal with these networks is that the simplified 
  polynomials that we get to describe the steady states in a given stoichiometric
compatibility class depend on a number of variables that grows linearly
with $n$ and the corresponding coefficient matrix does not have generic entries.
 We are nevertheless able to detect two simplices in these 
high dimensional spaces which share a facet, which are positively decorated
by the (huge) coefficient matrix.
 
 We first set the notation. 
 
 \subsection{Our setting}
 Using the notation in Figure~\ref{fig:cascaden}, we call $S_i^0, S_i^1$  the substrate proteins
  in the $i$-th layer, for $i=1,\dots,n$. As before, the upper index can be interpreted as the absence
   ($0$) or the presence ($1$) of a phosphate group in the substrate. The phosphorylation in the first layer is 
   catalyzed by the enzyme $S^1_0$. The activated protein $S^1_i$ in the $i$-th layer acts as the modifier 
   enzyme in the $(i+1)$-th layer. 
   The dephosphorylation in the $i$-th layer is carried out by a phosphatase $F_i$. 
   Some of the $F_i$ can be the same species, that is, the same phosphatase can react at different layers.
 
 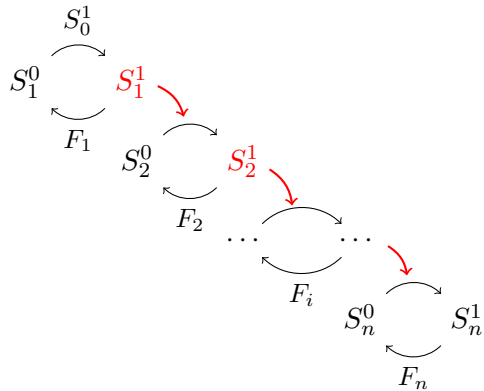
\begin{figure}[h]
 \centering
 \begin{tikzpicture}[scale=0.7,node distance=0.7cm] 
  \node[] at (5.3,-4.5) (dummy4) {};
  \node[left=of dummy4] (n0) {$S^0_n$};
  \node[right=of n0] (n1){$S^1_n$}
    edge[->, bend left=45] node[below] {\textcolor{black}{\small{$F_n$}}} (n0)
    edge[<-, bend right=45] node[above] {} (n0);
  \node[] at (3.2,-3) (dummy3) {};
  \node[left=of dummy3] (r0) {$\cdots$};
  \node[right=of r0] (r1){$\cdots$}
    edge[->, bend left=45] node[below] {\textcolor{black}{\small{$F_i$}}} (r0)
    edge[<-, bend right=45] node[above] {} (r0)
    edge[->,thick,color=red, bend left=25] node[above] {} ($(n0.north)+(25pt,10pt)$);
  \node[] at (1.1,-1.5) (dummy2) {};
  \node[left=of dummy2] (p0) {$S^0_2$};
  \node[right=of p0] (p1){\textcolor{red}{$S^1_2$}}
    edge[->, bend left=45] node[below] {\textcolor{black}{\small{$F_2$}}} (p0)
    edge[<-, bend right=45] node[above] {} (p0)
      edge[->,thick,color=red, bend left=25] node[above] {} ($(r0.north)+(25pt,10pt)$);
  \node[] at (-1,0) (dummy) {};
  \node[left=of dummy] (s0) {$S^0_1$};
  \node[right=of s0] (s1) {\textcolor{red}{$S^1_1$}}
    edge[->, bend left=45] node[below] {\textcolor{black}{\small{$F_1$}}} (s0)
    edge[<-, bend right=45] node[above] {\textcolor{black}{\small{$S^1_0$}}} (s0)
      edge[->,thick,color=red, bend left=25] node[above] {} ($(p0.north)+(25pt,10pt)$);

 \end{tikzpicture}
  \caption{Enzymatic cascade with $n$ layers.}
 \label{fig:cascaden}
 \end{figure}
 
  We assume the following reaction scheme:
 \[ \begin{array}{c} \label{eq:cascaden}
 %\nonumber 
% S^0_1 + E 
% \arrowschem{k_{\rm{on}_1}}{k_{\rm{off}_1}} Y^0_1
% \stackrel{k_{\rm{cat}_1}}{\rightarrow} S^1_1+E, \quad
 S^0_i + S^1_{i-1} \arrowschem{k_{\rm{on}_{i}}}{k_{\rm{off}_{i}}} Y^0_i 
 \stackrel{k_{\rm{cat}_{i}}}{\rightarrow} S^1_i + S^1_{i-1}, \qquad i=1\dots,n,\\
 
% S^1_i + F \arrowschem{\ell_{\rm{on}_{i}}}{\ell_{\rm{off}_{i}}} Y^1_i
% \stackrel{\ell_{\rm{cat}_{i}}}{\rightarrow} S^0_i+ F, \ i=1,2,
% 
% \quad
 S^1_i + F_i \arrowschem{\ell_{\rm{on}_{i}}}{\ell_{\rm{off}_{i}}} Y^1_i
 \stackrel{\ell_{\rm{cat}_{i}}}{\rightarrow} S^0_i+ F_i, \qquad i=1,\dots,n.
 
 \end{array}
 \]
 
We denote by $\mathcal{F}=\{P_1,\dots,P_r\}$ the set of phosphatases that appear in the network.
 In this case we have $4n+r+1$ chemical species: $S^1_0$, $S^0_1$, $S^1_1$, $S^0_2$, $S^1_2$,
 $\dots$,$S^0_n$, $S^1_n,$, $P_1$, $P_2 \dots P_r$, $Y^0_1$, $Y^1_1$, $Y^0_2$, $Y^1_2$,
 $\dots$, $Y^0_n$, $Y^1_n$. We denote the concentration of the species with small letters.
%
%
% 
% We denote by $s^0_1$, $s^1_1$, $s^0_2$, $s^1_2$,$\dots$,$s^0_n$, $s^1_n$  $e$, $f$, $y^0_1$, $y^1_1$, $y^0_2$, $y^1_2$,$\dots$,$y^0_n$,$y^1_n$ the concentration of 
% the $4n + 2$ chemical species in the network $S^0_1$, $S^1_1$, $S^0_2$, $S^1_2$,$\dots$,$S^0_n$, $S^1_n$  $E$, $F$, $Y^0_1$, $Y^1_1$, $Y^0_2$, $Y^1_2$,$\dots$,$Y^0_n$,$Y^1_n$
% respectively.
For each $j=1,\dots,r$, we call $\Lambda_j=\{i\in\{1,\dots,n\} \, : \,   F_i=P_j\}$
and we consider the function $j\colon\{1,\dots,n\}\to\{1,\dots,r\}$, defined by $j(i)= j$ if  $F_i=P_j$.

 The associated dynamical system that arises under mass-action kinetics is equal to:
 \vskip -10pt
 \begin{small}
 \begin{align*}
% \frac{ds^0_1}{dt} &=  {-k_{\rm{on}_1}}s^0_1e + {k_{\rm{off}_1}}y^0_1 + {\ell_{\rm{cat}_1}} y^1_1, \qquad
 \frac{ds^0_i}{dt} &=  {-k_{\rm{on}_i}}s^0_is^1_{i-1} + {k_{\rm{off}_i}}y^0_i + {\ell_{\rm{cat}_i}} y^1_i, \quad i=1,\dots,n,\\
% \frac{ds^1_i}{dt} &=  {k_{\rm{cat}_i}}y^0_i - {\ell_{\rm{on}_{i}}}s^1_if + {\ell_{\rm{off}_{1}}}y^1_i  -{k_{\rm{on}_{i+1}}}s^0_{i+1}s^1_i + 
% ({k_{\rm{off}_{i+1}}} + {k_{\rm{cat}_{i+1}}})y^0_{i+1}, \ i=1,2,\\
  \frac{ds^1_i}{dt} &=  {k_{\rm{cat}_i}}y^0_i - {\ell_{\rm{on}_{i}}}s^1_ip_{j(i)} + 
  {\ell_{\rm{off}_{1}}}y^1_i  -{k_{\rm{on}_{i+1}}}s^0_{i+1}s^1_i + 
 ({k_{\rm{off}_{i+1}}} + {k_{\rm{cat}_{i+1}}})y^0_{i+1}, \quad i=1,\dots,n-1,\\
 \frac{ds^1_n}{dt} &=  {k_{\rm{cat}_n}}y^0_n - {\ell_{\rm{on}_{n}}}s^1_np_{j(n)} +
  {\ell_{\rm{off}_{n}}}y^1_n,\\
% \frac{dy^0_1}{dt} &= {k_{\rm{on}_1}}s^0_1e  -{(k_{\rm{off}_1}+k_{\rm{cat}_1})}y^0_1,\qquad
 \frac{dy^0_i}{dt} &= {k_{\rm{on}_i}}s^0_is^1_{i-1}  -{(k_{\rm{off}_i}+k_{\rm{cat}_i})}y^0_i, \quad i=1,\dots,n,\\
 \frac{dy^1_i}{dt} &=  {\ell_{\rm{on}_i}}s^1_i p_{j(i)} -{(\ell_{\rm{off}_i}+\ell_{\rm{cat}_i})}y^1_i, \quad i=1,\dots,n,\\
 \frac{ds^1_0}{dt}&=-\frac{dy^0_1}{dt}, 
% \quad \frac{df}{dt}=-\sum_{i\in \Lambda_1} \frac{dy^1_i}{dt}, 
 \qquad \frac{dp_j}{dt}=-\sum_{i\in \Lambda_j} \frac{dy^1_i}{dt}, \quad j=1,\dots,r.
 \end{align*}
 \end{small}
 
 The space of linear forms yielding conservation laws has dimension $n+r+1$, 
 and we consider the following $n+r+1$ linearly independent conservation relations:
 \begin{align}\label{conscascadan}
 s^1_0 + y^0_1 = & S_{0,tot},\nonumber \\
 %f +  \sum_{i\in \Lambda_1} y^1_i = & F_{tot},\\
 s^0_i + s^1_i + y^0_i + y^1_i + y^0_{i+1} = & S_{i,tot},\quad i=1,\dots,n-1, \\
 s^0_n + s^1_n + y^0_n + y^1_n = & S_{n,tot},\nonumber\\
 p_j + \sum_{i\in \Lambda_j}  y^1_i = & P_{j,tot},\quad j=1,\dots,r.\nonumber
 \end{align}

 Again, following the general procedure described in \cite{aliciaMer}, we can find binomial equations
  that describe the concentration of the species at steady state. 
  The concentration of the intermediate species satisfy these binomial equations:
 \begin{align*}
 &y^0_i-K_i\,s^1_{i-1}s^0_i =0, \ \ i=1\dots,n,\qquad y^1_i-L_i\,p_{j(i)}s^1_i=0, \ \ i=1\dots,n,
 \end{align*}
 where $K_i=\frac{k_{\rm{on}_i}}{k_{\rm{off}_i}+k_{\rm{cat}_i}}$, $i=1,\dots,n$,  
 $L_i=\frac{\ell_{\rm{on}_i}}{\ell_{\rm{off}_i}+\ell_{\rm{cat}_i}}$, $i=1,\dots,n$. 
 The remaining binomials can be (algorithmically) chosen to be:
% \tau_1\,s^0_1e - \nu_1\,s^1_1f=0, \quad
 \[ \tau_i\,s^0_is^1_{i-1} - \nu_i\,s^1_ip_{j(i)}=0,\ i=1,\dots,n,\]
 where $\tau_i=k_{\rm{cat}_i}K_i$, $\nu_i=\ell_{\rm{cat}_i}L_i$, $i=1,\dots,n$.
 
As in the previous case of two layers, we can parametrize  the positive steady states by monomials. 
For instance, we can write  the concentrations of all species in terms of $s^1_i$,
 for $i=0, 1,\dots,n$ and $p_1, \dots,p_r$:
 \[\begin{array}{lrl}
 \label{parametrizationcascadefcualquiera}
s^0_i \, = & G_i \frac{s^1_ip_{j(i)}}{s^1_{i-1}},& i=1,\dots,n, \\
y^0_i \, =& K_iG_i\, s^1_ip_{j(i)}, & i=1,\dots,n,\nonumber\\
y^1_i  \, =& L_i\, s^1_ip_{j(i)}, & i=1,\dots,n,\nonumber
\end{array}\]
 where $G_i=\frac{\nu_i}{\tau_i}$ for all $i=1,\dots,n$.
 
 As in~\eqref{eq:A12}, we will denote for any $j=1, \dots, n$:
  \begin{equation}\label{eq:Aj}
A_j =\frac{\ell_{\rm{cat}_j}}{k_{\rm{cat}_j}}.
\end{equation}

\subsection{Statement of our main results}

Suppose first that there are two consecutive layers $i_0$, $i_0+1$,  $1\leq i_0\leq n-1$, with the same 
phosphatase $F$, that is, $P_{j(i_0)}=P_{j(i_0+1)}$, and with no restriction in the other layers.
Let $\alpha_{1,i_0},\alpha_{2,i_0},\alpha_{3,i_0}$ and $\alpha_{4,i_0}$ be as in the case $n=2$, but 
 these constants correspond to the restriction to the two layers $i_0$ and $i_0+1$. That is:
\begin{equation*}
\begin{aligned}
\alpha_{1,i_0}=&\dfrac{S_{i_0,tot}}{F_{tot}}- A_{i_0+1},\\
\alpha_{2,i_0}=&(A_{i_0}+1) - \dfrac{S_{i_0,tot}}{F_{tot}},\\
\alpha_{3,i_0}=&\dfrac{A_{i_0} +1-A_{i_0+1}}{A_{i_0}}\, \dfrac{S_{{i_0}-1,tot}}{F_{tot}} - 
 \left(\dfrac{S_{{i_0},tot}}{F_{tot}} - A_{{i_0}+1}\right),\\
\alpha_{4,{i_0}}=&
\dfrac{A_{i_0}+1-A_{{i_0}+1}}{A_{{i_0}+1}+1}\, \dfrac{S_{{i_0}+1,tot}}{F_{tot}} - 
 \left(A_{i_0}+1-\dfrac{S_{{i_0},tot}}{F_{tot}}\right),
\end{aligned}
\end{equation*}
%\begin{equation*} 
%\begin{aligned}
%\alpha_{1,i}=&S_{i,tot}- A_{i+1} \, F_{tot}, \\
%\alpha_{2,i}=& (A_i+1) \, F_{tot}- S_{i,tot},\\
%\alpha_{3,i}=& A_i \, A_{i+1} \, F_{tot} + (A_i
%+ 1 - A_{i+1}) \, S_{i-1,tot} -
%A_i \, S_{i,tot},\\
%\alpha_{4,i}=& (A_{i+1} + 1)S_{i,tot} -
% (A_i + 1)(A_{i+1} + 1)F_{tot} -  
%(A_{i+1} -A_i- 1) \, S_{{i+1},tot},
%\end{aligned}
%\end{equation*}
where the value of $E_{tot}$ in the case $n=2$ now corresponds to the value 
$S_{{i_0}-1,tot}$ and $F_{tot}=P_{j({i_0}),tot}=P_{j({i_0}+1),tot}$.
 We have the following result:
 
 \begin{thm}\label{th:provisorio-cascaden} Suppose $n\geq 3$, and suppose that there are two 
 consecutive layers ${i_0}$, ${i_0}+1$, with $1\leq {i_0}\leq n-1$, with the same phosphatase  and with no 
 restriction in the other layers. Let $A_{i_0}, A_{{i_0}+1}$ be as in~\eqref{eq:Aj}. 
 	Assume that the reaction rate constants verify 
 	\begin{equation}\label{eq:ineq}
 	A_{i_0}+1>A_{{i_0}+1}
 	\end{equation}
 	and the total 
	concentration constants verify the  inequalities
	$ \alpha_{1,{i_0}}, \alpha_{2,{i_0}},  \alpha_{3,{i_0}}, 	\alpha_{4,{i_0}}>0,$
	that is:
		$$A_{i_0}+1 > \dfrac{S_{{i_0},tot}}{F_{tot}} > A_{{i_0}+1}, \quad 
	\dfrac{S_{{i_0}-1,tot}}{F_{tot}} > \left(\dfrac{S_{{i_0},tot}}{F_{tot}}- A_{{i_0}+1} \right) \, \dfrac{A_{i_0}}{A_{i_0}+1-A_{{i_0}+1}},$$ %\quad 
$$\dfrac{S_{{i_0}+1,tot}}{F_{tot}} > \left(A_{i_0}+1- \dfrac{S_{{i_0},tot}}{F_{tot}}\right) \, \dfrac{A_{{i_0}+1} +1}{A_{i_0}+1-A_{{i_0}+1} },$$
% Assume the reaction rate constants and the total concentration 
%  constants verify one of the following sets of inequalities:
% \[\frac{\ell_{\rm{cat}_i}}{k_{\rm{cat}_i}} + 1 > \frac{\ell_{\rm{cat}_{i+1}}}
% {k_{\rm{cat}_{i+1}}}, \ \alpha_{1,i}, \alpha_{2,i}, \alpha_{3,i}, 
% \alpha_{4,i}>0$$ 
or $$A_{i_0} + 1 <
A_{{i_0}+1},\ \alpha_{1,{i_0}}, \alpha_{2,{i_0}},   \alpha_{3,{i_0}}, \alpha_{4,{i_0}}<0.$$

Then, there exists a rescaling in the constants $k_{on_i}, i=1,\dots,n$ and $\ell_{on_i}$, $i=1,\dots,n$, 
such that the system has at least two positive steady states.
  
%Fix generic positive numbers $h_2, h_3, h_7, h_8$, $h_{i}$, $i=10,\dots,2n+5$, such that $h_8<h_2$, 
%$h_8<h_{i+7}$, for  $i\in \Lambda_1, i \neq 1, 2$, $h_8<h_{i+n+5}$,  for $i\in \Lambda_1, i \neq 1, 2, 3$. 
%Then, there exists $t_0>0$ such that for any value $t\in(0,t_0)$ the system has at least two positive steady states after 
%modifying the coefficients $k_{on_1}$,$k_{on_2}$, $k_{on_i}$, $i=3,\dots,n$, $\ell_{on_1}$, $\ell_{on_2}$, 
%$\ell_{on_i}$, $i=3,\dots,n$ via the rescaling  $t^{-h_7}k_{on_1}$, $t^{-h_3-h_8}k_{on_2}$,$t^{h_{i+7}-h_{i+n+5}}k_{on_i}$, $i=3,\dots,n$, 
% $t^{-h_2-h_3}\ell_{on_1}$, $t^{-h_2}\ell_{on_2}$, $t^{h_{i+7}-h_2}\ell_{on_i}$ if $i \in \Lambda_1$ or  $t^{h_{i+7}}
%\ell_{on_i}$ if $i\not  \in \Lambda_1$  for $i=3,\dots,n$.

 \end{thm}
 
We will give an explicit rescaling in the proof.

\begin{remark} \rm 
In the statement of Theorem~\ref{th:provisorio-cascaden} 
we have conditions which are similar to those in the case $n=2$,
but depending  on the reaction rate 
constants corresponding to the layers ${i_0}$ and ${i_0}+1$ and total conservation constants.
Again, the two sets of inequalities in the statement of 
Theorem~\ref{th:provisorio-cascaden} are clearly compatible.

For $n\geq 3$, there is not only an increase in the number of variables but also in the number of conservation laws. 
 The idea of the proof of  Theorem~\ref{th:provisorio-cascaden} is to extend the simplices that appear in 
 the proof of Theorem~\ref{th:cascade2}  to simplices in the higher dimensional space, showing that 
 in fact the conditions of the new simplices to be positively decorated are basically the same.
\end{remark}

The other case is when the  layers which share a phosphatase are not consecutive. 
Assume $i_1 < i_2$ are two non-consecutive layers
sharing the same phosphatase. Assume also that
there are no two other consecutives
 layers with a common phosphatase between them (otherwise, we would be  in the hyphothesis of the previous case). That is,
 there exists $i_1, i_2$, with $1\leq i_1 < i_1+1 < i_2 \leq n$, such that $P_{j(i_1)}=P_{j(i_2)}=F$, 
 and $P_{j(i)}$ for $i=i_1 + 1,\dots, i_2 -1$ are all distinct and different from $F$. We impose no restrictions 
 on the phosphatases of the remaining layers layers $1,\dots,i_1 - 1, i_2+1,\dots,n$.

Consider the following rational functions $\beta_{1,i_1,i_2}$, $\beta_{2,i_1,i_2}$, $\beta_{3,i_1,i_2}$ 
and $\beta_{4,i_1,i_2}$ depending on the catalytic reaction rate constant and total concentration constants:
\begin{equation*}
\begin{aligned}
\beta_{1,i_1,i_2}=& \dfrac{S_{i_1-1,tot}}{S_{i_1,tot}} -
 \dfrac{A_{i_1}}{A_{i_1+1}},\\
\beta_{2,i_1,i_2}=&(A_{i_1}+1) - \dfrac{S_{{i_1},tot}}{ F_{tot}},\\
\beta_{3,i_1,i_2}=&\dfrac{S_{i_2-1,tot}}{F_{tot}} -\left( A_{i_1}+1\right)\,  \dfrac{S_{i_2,tot}}{F_{tot}},\\
\beta_{4,i_1,i_2}=&\dfrac{S_{{i_1},tot}}{F_{tot}} - 
\left(\dfrac{A_{i_1}+1}{A_{i_2}+1}\right) \, \left(A_{i_2}+1-
\dfrac{ S_{{i_2},tot}}{F_{tot}}\right),
\end{aligned}
\end{equation*}
where $F_{tot}=P_{j(i_1),tot}=P_{j(i_2),tot}$.

We then have:

 \begin{thm}\label{th:cascadan-noconsecutivos} Suppose $n\geq 3$, and suppose there exists
  layers $i_1, i_2$, with $1\leq i_1 < i_1+1 < i_2 \leq n$, such that $P_{j(i_1)}=P_{j(i_2)}=F$, 
  $P_{j(i)}$ for $i=i_1 + 1,\dots, i_2 -1$ are all distinct and different from $F$, and with no 
  restriction in the phosphatases of layers $1,\dots,i_1 - 1, i_2+1,\dots,n$. Assume the reaction rate 
  constants and the total concentration constants verify 
\[\beta_{1,i_1,i_2},\, \beta_{2,i_1,i_2},\, \beta_{3,i_1,i_2},\, \beta_{4,i_1,i_2} >0.\]
Then, there exists a rescaling in the constants $k_{on_i}, i=1,\dots,n$ and $\ell_{on_i}$, $i=1,\dots,n$, 
such that the system has at least two positive steady states.
  
%Fix generic positive numbers $h_2, h_3, h_7, h_8$, $h_{i}$, $i=10,\dots,2n+5$, such that $h_8<h_2$, $h_8<h_{i+7}$, 
%for  $i\in \Lambda_1, i \neq 1, 2$, $h_8<h_{i+n+5}$,  for $i\in \Lambda_1, i \neq 1, 2, 3$. 
%Then, there exists $t_0>0$ such that for any value $t\in(0,t_0)$ the system has at least 
%two positive steady states after modifying the coefficients $k_{on_1}$,$k_{on_2}$, 
%$k_{on_i}$, $i=3,\dots,n$, $\ell_{on_1}$, $\ell_{on_2}$, $\ell_{on_i}$, $i=3,\dots,n$ 
%via the rescaling  $t^{-h_7}k_{on_1}$, $t^{-h_3-h_8}k_{on_2}$,$t^{h_{i+7}-h_{i+n+5}}k_{on_i}$, $i=3,\dots,n$, 
% $t^{-h_2-h_3}\ell_{on_1}$, $t^{-h_2}\ell_{on_2}$, $t^{h_{i+7}-h_2}\ell_{on_i}$ if 
% $i \in \Lambda_1$ or  $t^{h_{i+7}}\ell_{on_i}$ if $i\not  \in \Lambda_1$  for $i=3,\dots,n$.

\end{thm}

Again, we will give an explicit rescaling in the proof.

\begin{remark} \rm The inequalities in the statement of Theorem~\ref{th:cascadan-noconsecutivos} 
are compatible. They have a similar flavour, but they are different from the conditions defining the
regions of multistationarity in Theorems~\ref{th:cascade2} and~\ref{th:provisorio-cascaden}.
\end{remark}

\subsection{The proof of Theorem~\ref{th:provisorio-cascaden}}

 \begin{proof}[Proof of Theorem \ref{th:provisorio-cascaden}] 
Without loss of generality we suppose that the phosphatase in the layers ${i_0}$ and ${i_0}+1$
 is the phosphatase $P_1$, that we call $F$. We showed in (\ref{parametrizationcascadefcualquiera})
  that we can parametrize the steady states in terms of the concentrations of $s^1_i$, 
  for $i=0,\dots,n$, $f$ (we use $f$ instead of $p_1$) and $p_i$, for $i=2,\dots,r$. 
  To avoid unnecessary notation, in this proof we call $s_i=s^1_i$ for all $i=0,\dots,n$.

Consider the following set of monomials:
\[\mathcal{M}=\{s_{{i_0}-1},f,s_{{i_0}},s_{{i_0}+1},
s_{i_0}f,s_{{i_0}+1}f,s_{{i_0}}f(s_{{i_0}-1})^{-1},s_{{i_0}+1}f(s_{{i_0}})^{-1},1\}.\]
These monomials appear in the parametrization of the concentration at steady state of 
the species in layers ${i_0}$ and ${i_0}+1$. % (the corresponding monomials of the case with $n=2$ layers).
Now, consider the set
\[\mathcal{M'}=\mathcal{M}\cup \{s_{0},s_{1},\dots,s_{{i_0}-2},s_{{i_0}+2},\dots,s_{n},p_2,\dots,p_r\}.\]
And consider also the set of all the monomials that appear in the parametrization:
\begin{align*}
\mathcal{M''}=\mathcal{M'}\cup &\{s_1p_{j(1)},\dots,s_{{i_0}-1}p_{j({i_0}-1)},s_{{i_0}+2}p_{j({i_0}+2)},
\dots,s_np_{j(n)},s_{1}p_{j(1)}(s_{0})^{-1},\dots\\
&\dots,s_{{i_0}-1}p_{j({i_0}-1)}(s_{{i_0}-2})^{-1},s_{{i_0}+2}p_{j({i_0}+2)}(s_{{i_0}+1})^{-1},\dots,s_{n}p_{j(n)}(s_{n-1})^{-1}\}.
\end{align*}
We have $n+r+1$ variables: $s_0,s_1,\dots,s_n,f,p_2,\dots,p_r$. Consider the variables 
with this last order. Let $\mathcal{A}, \mathcal{A'}$, $\mathcal{A''}\subset{\R^{n+r+1}}$ be
the subsets corresponding to the supports of the sets $\mathcal{M}, \mathcal{M'}$, $\mathcal{M''}$ 
respectively, that is, the exponents of the monomials in each set.

We consider an order in $\mathcal{A''}$ given by the order in which 
we construct $\mathcal{M''}$: first the exponents corresponding to monomials in $\mathcal{M}$ 
(in that order), then the exponents corresponding to monomials that we add to obtain $\mathcal{M'}$
(in that order), and then the rest of the exponents, in the same order as enumerated above. 
 We have $3n+r+2$ monomials. 

As in the case $n=2$, we replace the monomial parametrization into the conservation laws~\eqref{conscascadan} 
and we write this system in a matricial form.  
We call $C\in \R^{(n+r+1)\times(3n+r+2)}$ the matrix of coefficients of the resulting polynomial system.

We want to find two simplices with vertices in $\mathcal{A''}$ which share a facet. Inspired by the $4$-simplices 
that we chooses for the case $n=2$, we take the following $(n+r+1)$-simplices:
\begin{align*}
\Delta_1&=\{e_{i_0},\, e_{{i_0}+1}+e_{n+2},\, e_{{i_0}+2}+e_{n+2},\, e_{{i_0}+2},\, 0\}
\cup(\mathcal{A'}\setminus \mathcal{A}),\\
\Delta_2&=\{e_{i_0},\, e_{{i_0}+1}+e_{n+2},\, e_{{i_0}+2}+e_{n+2},\, e_{{i_0}+2}+e_{n+2}-e_{{i_0}+1},\, 0\}
\cup(\mathcal{A'}\setminus \mathcal{A}).
\end{align*}
where $e_i$ denotes the $i$-th canonical vector of $\R^{n+r+1}$. Note that the points $e_{i_0}$, $e_{{i_0}+2}$,
 $e_{{i_0}+1}+e_{n+2}$, $e_{{i_0}+2}+e_{n+2}$, $0$ correspond to the monomials $s_{{i_0}-1}$,  $s_{{i_0}+1}$, 
 $s_{{i_0}}f$, $s_{{i_0}+1}f$, $1$, and the points $e_{i_0}$, $e_{{i_0}+1}+e_{n+2}$ ,$e_{{i_0}+2}+e_{n+2}$,
  $e_{{i_0}+2}+e_{n+2}-e_{{i_0}+1}$, $0$, correspond to the monomials $s_{{i_0}-1}$, $s_{{i_0}}f$,
   $s_{{i_0}+1}f$, $s_{{i_0}+1}f(s_{i_0})^{-1}$, $1$ which are in correspondence 
   with the points of the simplices in the proof of Theorem~\ref{th:cascade2}.

We consider first the equations corresponding to the conservation laws with total conservation 
constants $S_{{i_0}-1,tot}$, $F_{tot}$, $S_{{i_0},tot}$, $S_{{i_0}+1,tot}$ and then the equations corresponding 
to the conservation constants $S_{0,tot}$, $\dots$, $S_{{i_0}-2,tot,}$, $S_{{i_0}+2,tot}$,$\dots$, $S_{n,tot}$,
 $P_{2,tot}$, $\dots$, $P_{n,tot}$. The submatrices of  $C$ restricted to the columns
  corresponding to the simplex $\Delta_j$, for $j=1,2$, are equal to:
\begin{small}
\[C_{\Delta_j} =\left( \begin{array}{cccc|ccc}
& & C_j & & & 0 & \\ \hline
0 & \dots & 0 & -S_{0,tot}& &  &  \\
\vdots & \ddots &\vdots & \vdots &  &  & \\
0 & \dots & 0 & -S_{{i_0}-2,tot}& &  &  \\
0 & \dots & 0 & -S_{{i_0}+2,tot}& &  &  \\
\vdots & \ddots &\vdots & \vdots &  & \Id_{n+r-3}  & \\
0 & \dots & 0 & -S_{n,tot}& & & \\
0 & \dots & 0 & -P_{2,tot}& & & \\
\vdots & \ddots &\vdots & \vdots &  &  & \\
0 & \dots & 0 & -P_{r,tot}& & & \\
\end{array}\right),\]
\end{small}
where $C_1$ is the submatrix corresponding to columns of the exponents
 $\{e_{i_0},e_{{i_0}+1}+e_{n+2},e_{{i_0}+2}+e_{n+2},e_{{i_0}+2},0\}$ and $C_2$ is the submatrix corresponding
  to the columns of the exponents $\{e_{i_0},e_{{i_0}+1}+e_{n+2},e_{{i_0}+2}+e_{n+2},e_{{i_0}+2}+e_{n+2}-e_{{i_0}+1},0\}$, that is:
	\begin{equation*}\label{matrixc1}
	C_j=\begin{pmatrix}
	1 & K_{i_0}G_{i_0} & 0 & 0 & -S_{{i_0}-1, tot}\\
	0  & L_{i_0} & L_{{i_0}+1} & 0 & -F_{tot}\\
	0  & K_{i_0}G_{i_0} + L_{i_0} & K_{{i_0}+1}G_{{i_0}+1} & 0 & -S_{{i_0},tot} \\
	0  & 0 & K_{{i_0}+1}G_{{i_0}+1} + L_{{i_0}+1} & (C_j)_{44} & -S_{{i_0}+1,tot}
	\end{pmatrix},
	\end{equation*}
where $(C_1)_{44}=1$ and $(C_2)_{44}=G_{{i_0}+1}$.

Note that the matrix $C_{\Delta_j}$ is positively spanning if and only if $C_j$ is positively spanning, 
for $j=1,2$. Moreover, $C_1$ and $C_2$ are positively spanning if and only if the condition of the statement holds,
that is,  $A_{i_0} + 1 > A_{{i_0}+1}$
and $\alpha_{1,{i_0}}, \alpha_{2,{i_0}}, \alpha_{3,{i_0}}, \alpha_{4,{i_0}} >0$; or 
$A_{i_0} + 1 < A_{{i_0}+1}$
 and $\alpha_{1,{i_0}}, \alpha_{2,{i_0}}, \alpha_{3,{i_0}}, \alpha_{4,{i_0}} <0$.

Given $h\in\mathcal{C}_{\Delta_1,\Delta_2}$, by Theorem~\ref{th:BS}
there exists $t_0 \in \R_{>0}$ such that for all $0<t<t_0$, 
the number of positive (nondegenerate) solutions of the scaled system, that is, the system with the 
same support $\mathcal{A''}$ and matrix of coefficients $C_t$ with $(C_t)_{ij}=t^{h(\alpha_j)}c_{ij}$, 
with $\alpha_j\in \mathcal{A''}$ and $C=(c_{ij})$, is at least two. 
% , $h_2=h(e_2)$, $h_3=h(e_3)$, $h_4=h(e_4)$, $h_5=h(e_2+e_3)$, $h_6=h(e_2+e_4)$,
%  $h_7=h(-e_1+e_2+e_3)$, $h_8=h(e_2-e_3+e_4)$, $h_9=h(0)$,  
%  $h_{i+5}=h(e_{i})$ for $i=5,\dots,n+2$,
%   $h_{i+n+3}=h(e_2+e_{i})$ for $i=5,\dots,n+2$, $h_{i+2n+1}=h(e_2+e_{i}-e_{i-1})$ for $i=5,\dots,n+2$. 
This system has the following form. Call $x=(s_0,s_1,\dots,s_n,f,p_2,\dots,p_r)$ and note that each
coefficient $c_{ij}$ is a rational function of the vector of reaction rate constants that we call $\kappa=(k_{on_1}, \ell_{on_1}, \dots,)$. 
To emphasize this,
we write $c_{ij}=c_{ij}(\kappa)$. Moreover, setting $\gamma_j=t^{h(\alpha_j)}$ for any $j$, we have
a Laurent polynomial system of $n+r+1$ equations in $n+r+1$ variables:
\begin{equation}\label{eq:5.4}
\sum_j c_{ij}(\kappa) \, \gamma_j \, x^{\alpha_j} \, = \, 0, \, i=1, \dots, n+r+1.
\end{equation}
Now, the reaction network we are considering satisfies the hypotheses of Theorem~5.4 of \cite{AGB1}. 
 Then, there exists a vector of reaction rate 
constants $\bar \kappa$ such that the number of positive solutions of system~\eqref{eq:5.4} coincides
with the number of positive solutions of the following system:
\begin{equation}\label{eq:5.5}
\sum_j c_{ij}(\bar \kappa) \, x^{\alpha_j} \, = \, 0, \, i=1, \dots, n+r+1.
\end{equation}
We now describe the associated vector $\bar \kappa$ in an explicit form.
We first describe the cone $\mathcal{C}_{\Delta_1,\Delta_2}$ defined in~\eqref{eq:Delta12}.
We denote by $h_j$ the height corresponding to $\alpha_j\in\mathcal{A''}$, for $j=1,\dots,3n+r+2$ (in the 
order corresponding to the construction of $\mathcal M''$ that we described before). 
Let $\varphi_1$ and $\varphi_2$ be the affine linear functions  
 which agree with $h$ on the simplices $\Delta_1$ and $\Delta_2$ respectively. 
 We can take the heights of the points of $\Delta_1$ 
as zero, that is, $h_1=h_4=h_5=h_6=h_9=0$ and $h_j=0$ for $j=10,\dots,n+r+1$, and $h_8>0$
 (the height of the remaining point of $\Delta_2$ which is not in $\Delta_1$). Then,
 $\varphi_1(x_1,\dots,x_{n+r+1})=0$ and $\varphi_2(x_1,\dots,x_{n+r+1})=- h_8\, x_{{i_0}+1} - h_8\, x_{{i_0}+2}+h_8\, x_{n+2}$.
 Moreover, $h$ satisfies $0=\varphi_1(\alpha)<h(\alpha)$, for all $\alpha\not\in \Delta_1$ and 
 $\varphi_2(\alpha)<h(\alpha)$, for all $\alpha\not\in \Delta_1, \Delta_2$. Then, we have:
\vskip -10pt
\begin{align*}
&h_8 < h_2,\ \quad 0 < h_3, \quad 0 < h_7,\\
&h_j >0,\ \mbox{for} \ j=n+r+7,\dots,3n+r+2, \\
&h_8 < h_{n+r+6+j},\ \mbox{for} \ j\in\Lambda_1, \mbox{with}\ j= 1,\dots,{i_0}-1, \\
&h_8 < h_{n+r+4+j},\ \mbox{for} \ j\in\Lambda_1, \mbox{with}\ j= {i_0}+2,\dots,n, \\
&h_8 < h_{2n+r+4+j}\ \mbox{for}\ j\in\Lambda_1, \mbox{with}\  j=1,\dots,{i_0}-1,\\
&2h_8 < h_{2n+r+4+i},\ \mbox{if} \ {i_0}+2\in\Lambda_1, \mbox{and }\ h_8 < h_{2n+r+2+j}\
 \mbox{for}\ j\in\Lambda_1, \mbox{with}\  j={i_0}+3,\dots,n.
\end{align*}
where $h_2,h_3,h_7$ and $h_{j}$, $j=n+r+7,\dots,3n+r+2$, are generic.
  
If we change the variables $\bar f = t^{h_2}f$, $\bar s_{i_0} = t^{h_3}s_{i_0}$, we consider the constants:

\vskip -10pt
\begin{align}\label{scalingcascadegeneraln}
\overline{K_{i_0}}&=t^{-h_7}K_{i_0}, \quad
 \overline{K_{{i_0}+1}}=t^{-h_3-h_8}K_{{i_0}+1},\quad \overline{L_{i_0}}=t^{-h_2-h_3}L_{i_0},\quad
 \overline{L_{{i_0}+1}}=t^{-h_2}L_{{i_0}+1},  \nonumber\\
 \overline{K_i}&=t^{h_{n+r+6+i}-h_{2n+r+4+i}}K_i, i=1,\dots,{i_0}-1, \nonumber\\ 
 \overline{K_i}&=t^{h_{n+r+4+i}-h_{2n+r+2+i}}K_i, i={i_0}+2,\dots,n,\\
  \overline{L_i}&=t^{h_{n+r+6+i}-h_2}L_i,\ \mbox{if} \ i\in\Lambda_1, \
   \overline{L_i}=t^{h_{n+r+6+i}}L_i,\ \mbox{if} \ i\not \in\Lambda_1, \ \text{for } i=1,\dots,{i_0}-1, \nonumber \\
    \overline{L_i}&=t^{h_{n+r+4+i}-h_2}L_i,\ \mbox{if} \ i\in\Lambda_1, \ 
    \overline{L_i}=t^{h_{n+r+4+i}}L_i,\ \mbox{if} \ i\not \in\Lambda_1, \ \text{for } i={i_0}+2,\dots,n, \nonumber 
\end{align}
and we keep fixed the values of the constants $k_{\rm{cat}_{i_0}}$, $k_{\rm{cat}_{{i_0}+1}}$, ${\ell_{\rm{cat}_{i_0}}}$ 
 and ${\ell_{\rm{cat}_{{i_0}+1}}}$ and the total conservation constants $S_{{i_0}-1,tot}$, 
 $F_{tot}$, $S_{{i_0},tot}$ and $S_{{i_0}+1,tot}$, then the
  dynamical system associated with the network with these constants is system~\eqref{eq:5.5} with coefficients
  depending on the scaled reaction constants.
Therefore, the cascade we are considering has at least two positive steady states for these constants.

  To get the scalings in (\ref{scalingcascadegeneraln}) it can be checked that
   it is enough to rescale the original constants as follows:
  \begin{align*}
\overline{k_{on_{\rm {i_0}}}}&=t^{-h_7}{k_{on_{\rm {i_0}}}}, \quad
 \overline{{k_{on_{\rm {i_0}+1}}}}=t^{-h_3-h_8}{k_{on_{\rm {i_0}+1}}},\quad \overline{{\ell_{on_{\rm {i_0}}}}}=t^{-h_2-h_3}{\ell_{on_{\rm {i_0}}}},\quad
 \overline{{\ell_{on_{\rm {i_0}+1}}}}=t^{-h_2}{\ell_{on_{\rm {i_0}+1}}},  \nonumber\\
 \overline{{k_{on_{\rm i}}}}&=t^{h_{n+r+6+i}-h_{2n+r+4+i}}{k_{on_{\rm i}}}, i=1,\dots,{i_0}-1, 
 \nonumber\\ \overline{{k_{on_{\rm i}}}}&=t^{h_{n+r+4+i}-h_{2n+r+2+i}}{k_{on_{\rm i}}}, i={i_0}+2,\dots,n,\\
  \overline{{\ell_{on_{\rm i}}}}&=t^{h_{n+r+6+i}-h_2}{\ell_{on_{\rm i}}},\ \mbox{if} \ i\in\Lambda_1, 
  \ \overline{{\ell_{on_{\rm i}}}}=t^{h_{n+r+6+i}}{\ell_{on_{\rm i}}},\ \mbox{if} \ i\not \in\Lambda_1, \ \text{for } i=1,\dots,{i_0}-1, \nonumber \\
    \overline{{\ell_{on_{\rm i}}}}&=t^{h_{n+r+4+i}-h_2}{\ell_{on_{\rm i}}},\ \mbox{if} \ i\in\Lambda_1, 
    \ \overline{{\ell_{on_{\rm i}}}}=t^{h_{n+r+4+i}}{\ell_{on_{\rm i}}},\ \mbox{if} \ i\not \in\Lambda_1, \ \text{for } i={i_0}+2,\dots,n. \nonumber 
\end{align*}

 \end{proof}

  \subsection{The proof of Theorem \ref{th:cascadan-noconsecutivos}}
 
   For simplicity and to fix ideas, 
we only present a proof in the case when the first and the last layer have the same phosphatase (that is, $i_1=1, i_2=n$),
and the  other layers have all different phosphatases. We also present an explicit rescaling for this case. 
 The general case is similar, but with a heavier notation.

\begin{proof}[Proof of Theorem \ref{th:cascadan-noconsecutivos}] 
We call $f$ the concentration of the phosphatase $F$, the phosphatase that appear in the layers $1$ and $n$, 
and we call $f_i=p_{j(i)}$ for $i=2,\dots,n-1$. By assumption the variables $f_i$ are all distinct and different from $f$. 
We showed in \eqref{parametrizationcascadefcualquiera} that we can parametrize the steady states in terms of the
 concentrations of $s^1_i$, for $i=0,\dots,n$, $f$ and $f_i$, for $i=2,\dots,n-1$. To avoid unnecessary notation, 
 in this proof we call again $s_i=s^1_i$ for all $i=0,\dots,n$.

We have $2n$ variables, and we consider these $2n$ variables with the following order: $s_0$, $s_1,\dots,s_n$, 
$f$, $f_2,\dots, f_{n-1}$. In the monomial parametrization there are $4n+1$ different monomials, and we consider 
these monomials with this order: $s_0$, $s_1,\dots,s_n$, $f$, $f_2,\dots, f_{n-1}$, $s_1f$, $s_2f_2,
\dots,s_{n-1}f_{n-1}$, $s_nf$, $s_1f(s_0)^{-1}$, $s_2f_2(s_1)^{-1},\dots, s_{n-1}f_{n-1}(s_{n-2})^{-1}$, $s_nf(s_{n-1})^{-1}, 1$.

As in the previous cases, we replace the monomial parametrization into the conservation laws and we write 
this system in  matricial form. Let $C\in \R^{(2n)\times(4n+1)}$ be  the matrix of coefficients
 of this polynomial system. We call $\mathcal{A}$ the support of the system.

We want to find two simplices with vertices in $\mathcal{A}$ which share a facet.
 We denote $\mathcal{B}\subset\mathcal{A}$ the set of the exponents
  corresponding to the monomials:  $s_2f_2(s_1)^{-1}$,$\dots$, 
  $s_{n-2}f_{n-2}(s_{n-3})^{-1}$, $f_2,\dots,f_{n-1}$. 
  We consider the two following simplices: $\Delta_1$ given by the exponents corresponding to the monomials 
  $s_0$, $s_1f$, $s_nf$, $s_{n-1}f_{n-1}(s_{n-2})^{-1}$, $s_n$, $1$, 
  and the points in $\mathcal{B}$, and $\Delta_2$ given by the exponents corresponding to the monomials 
  $s_0$, $s_1f$, $s_nf$, $s_{n-1}f_{n-1}(s_{n-2})^{-1}$, $s_nf(s_{n-1})^{-1}$, $1$
   and the points in $\mathcal{B}$. That is:
\begin{align*}
\Delta_1&=\{e_1,e_2+e_{n+2},e_{n+1}+e_{n+2},e_{n}+e_{2n}-e_{n-1}, e_{n+1},0\}\cup\mathcal{B},\\
\Delta_2&=\{e_1,e_2+e_{n+2},e_{n+1}+e_{n+2},e_{n}+e_{2n}-e_{n-1}, e_{n+1}+e_{n+2}-e_n,0\}\cup\mathcal{B},
\end{align*}
where $e_i$ denotes the $i$-th canonical vector of $\R^{2n}$.

If we consider first the equations corresponding to the conservation laws with total 
conservation constants $S_{0,tot}$, $F_{tot}$, $S_{1,tot}$, $S_{n-1,tot}$, $S_{n,tot}$ 
and then the equations corresponding to the conservation constants 
$S_{2,tot},\dots,S_{n-2,tot}$, $F_{2,tot},\dots, F_{n-1,tot}$, the submatrices of 
$C$ restricted to the columns corresponding to the simplices $\Delta_j$ for $j=1,2$ 
 are as follows. We change the order of the columns following the order of the monomials in each simplex;
 the property of being positively spanning remains invariant: 
\[C_{\Delta_j} =\left( \begin{array}{cccc|ccc|ccc}
& & C_j & & & 0 & & & 0  &\\ \hline
0 & \dots & 0 & -S_{2,tot}& &  & & &  & \\
\vdots & \ddots &\vdots & \vdots &  & G  & & & 0  &\\
0 & \dots & 0 & -S_{n-2,tot}& &  &  & &  &\\ \hline
0 & \dots & 0 & -F_{2,tot}& & & & &  &\\
\vdots & \ddots &\vdots & \vdots &  & 0 &  & & \Id_{n-2}   &\\
0 & \dots & 0 & -F_{n-1,tot}& & & & &  &\\

\end{array}\right),\]
where $G\in\R^{(n-3)\times(n-3)}$ is the diagonal matrix with $G_{ii}=G_{i+1}$,
 for $i=1,\dots,n-3$, $C_1$ is the submatrix corresponding to columns of the exponents 
 $\{e_1,e_2+e_{n+2},e_{n+1}+e_{n+2},e_{n}+e_{2n}-e_{n-1}, e_{n+1},0\}$
and $C_2$ is the submatrix corresponding to the columns of the exponents 
$\{e_1,e_2+e_{n+2},e_{n+1}+e_{n+2},e_{n}+e_{2n}-e_{n-1}, e_{n+1}+e_{n+2}-e_n,0\}$:
	\begin{equation}
	C_j=\begin{pmatrix}
	1 & K_1G_1 & 0 & 0 & 0 & -S_{0, tot}\\
	0 & L_1 & L_n & 0 & 0 & -F_{tot}\\
	0 &  K_1G_1 + L_1 & 0 & 0 & 0 & -S_{1,tot} \\
	0 & 0 & K_nG_n & G_{n-1} & 0 &  -S_{n-1,tot} \\
	0 & 0 &  K_{n}G_{n} + L_{n}  & 0 & (C_j)_{55} & -S_{n,tot}
	\end{pmatrix},
	\end{equation}
with $(C_1)_{55}=1$ and $(C_2)_{55}=G_n$.

We observe that the matrix $C_{\Delta_j}$ if positively spanning if and only 
if $C_j$ is positively spanning, for $i=1,2$. It is straightforward to check that
the conditions  under which $C_1$ and $C_2$
 are positively spanning are equivalent to the conditions of the statement: $\beta_{1,1,n}$, 
 $\beta_{2,1,n}$, $\beta_{3,1,n}$, $\beta_{4,1,n}>0$.

Given $h\in\mathcal{C}_{\Delta_1,\Delta_2}$, by Theorem~\ref{th:BS}
 there exists $t_0\in\R_{>0}$ such that for all $0<t<t_0$, the number of 
 positive (nondegenerate) solutions of the scaled system, i.e. the system with support 
 $\mathcal{A}$ and matrix of coefficients $C_t$, with $(C_t)_{ij}=t^{h(\alpha_j)}c_{ij}$,
  (with $\alpha_j\in\mathcal{A}$, $C=(c_{ij})$) is at least two.
  
We now argue as in the proof of Theorem~\ref{th:provisorio-cascaden}. 
We can write our system in the form~\eqref{eq:5.4}, and 
since any cascade of Goldbeter-Koshland loops satisfies
the hypotheses of Theorem~5.4 of \cite{AGB1}, we again deduce
 the existence of a vector of rate constants $\bar \kappa$ such that the number of
 positive solutions of this system coincides with
 the number of positive solutions of  the corresponding system of the form~\eqref{eq:5.5}.
   In what follows, we also explicitly describe the rescaling of the parameters.

We denote  by $h_j$  the height corresponding to $\alpha_j\in\mathcal{A}$, 
for $j=1,\dots,4n+1$, with the order of the monomials as we described before. 
Let $\varphi_1$ and $\varphi_2$ be the affine linear functions which agree 
with $h$ on the simplices $\Delta_1$ and $\Delta_2$ respectively.

We can take zero heights at the points of $\Delta_1$, that is, $h_1=h_{2n+1}=
h_{3n}=h_{4n-1}=h_{n+1}=h_{4n+1}=0$, $h_j=0$
 for $j=n+3,\dots,2n$, $h_j=0$, for $j=3n+2,\dots,4n-2=0$ if $n>3$ (note that if $n=3$, 
$h_{3n+2}=h_{4n-1}$, already defined) and $h_{4n}>0$
(the height of the other point of $\Delta_2$). Then, 
 $\varphi_1(x_1,\dots,x_{2n})=0$ and $\varphi_2(x_1,\dots,x_{2n})=-h_{4n}\sum_{i=2}^n x_{i} + h_{4n}\, x_{n+2}$.
 
As $h\in\mathcal{C}_{\Delta_1,\Delta_2}$, we have that $h$ satisfies 
$0=\varphi_1(\alpha)<h(\alpha)$, for all $\alpha\not\in 
\Delta_1$ and $\varphi_2(\alpha)<h(\alpha)$, for all $\alpha\not\in \Delta_1, \Delta_2$. 
Then, we have these conditions:
\vskip -10pt
\begin{align*}
&h_{4n}<h_{n+2},\qquad h_j>0,\ \mbox{for} \ j=2,\dots,n, 2n+2,\dots, 3n-1, 3n+1,
\end{align*}
where $h_j$ for $j=2,\dots,n,n+2,2n+2,\dots,3n-1, 3n+1$ are generic.

If we change the variables $\bar s_i = t^{h_{i+1}}s_i$, for $i=1,\dots,n-1$ and 
$\bar f = t^{h_{n+2}}f$, we consider the constants:
\begin{equation}\label{scalingcascadegeneralnnoconsecutivas}
\begin{matrix}
\overline{K_1}&=&t^{-h_{3n+1}}K_1, &\quad&
 \overline{L_1}&=&t^{-h_2-h_{n+2}}L_1,& \\
 \overline{K_i}&=&t^{h_{2n+i}-h_i}K_i,&\quad & \overline{L_i}&=
 &t^{h_{2n+i}-h_{i+1}}L_i,&\quad \text{for } i=2,\dots,n-1,\\
\overline{K_n}&=&t^{-h_{n}-h_{4n}}K_n,& \quad
 &\overline{L_n}&=&t^{-h_{n+2}}L_n, &\quad
\end{matrix}
\end{equation}
without altering the values of the constants $k_{\rm{cat}_1}$, $k_{\rm{cat}_n}$, 
${\ell_{\rm{cat}_1}}$, ${\ell_{\rm{cat}_n}}$ and the total conservation values, then the
  dynamical system associated with the network with these constants is the scaled system.
Therefore,  the network has at least two positive steady states for this choice of constants.

It is straightforward to check that to get the scalings in (\ref{scalingcascadegeneralnnoconsecutivas})
 it is enough to rescale the original constants as follows:
\begin{equation*}
\begin{matrix}
\overline{k_{on_{\rm 1}}}&=&t^{-h_{3n+1}}k_{on_{\rm 1}}, &\quad&
 \overline{\ell_{on_{\rm 1}}}&=&t^{-h_2-h_{n+2}}\ell_{on_{\rm 1}},& \\
 \overline{k_{on_{\rm i}}}&=&t^{h_{2n+i}-h_i}k_{on_{\rm i}},&\quad &
  \overline{\ell_{on_{\rm i}}}&=&t^{h_{2n+i}-h_{i+1}}\ell_{on_{\rm i}},&\quad \text{for } i=2,\dots,n-1,\\
\overline{k_{on_{\rm n}}}&=&t^{-h_{n}-h_{4n}}k_{on_{\rm n}},& \quad
 &\overline{\ell_{on_{\rm n}}}&=&t^{-h_{n+2}}\ell_{on_{\rm n}}. &\quad
\end{matrix}
\end{equation*}

\end{proof}

 \section{Acknowledgment} The authors are grateful to the Kurt and Alice 
 Wallenberg Foundation and to the Institut Mittag-Leffler, Sweden, for their
 support to work on this project. We are also grateful to 
 the Mathematics Department of the Royal Institute of Technology, Sweden,
 for the wonderful hospitality we enjoyed, and to the French Program PREFALC and the University of Buenos Aires, which made possible the visit of F. Bihan.

\appendix

\section{General statements behind our results about cascades}

 In the proof of Theorem~\ref{th:provisorio-cascaden} we extrapolated the multistationarity behaviour 
 and the description of a region of multistationarity 
 of a subnetwork (described in~\ref{th:cascade2}) to the whole network, 
 even if it has more linearly independent conservation relations. For this, we 
 developed some ideas that we now abstract in Theorem~\ref{thm:extending} and 
 that can be used in the study of other cascade mechanisms. As we remarked at the
 end of the Introduction, they can be applied to describe a multistationarity
 region for the Ras cascade in Figure~\ref{fig:RAS} (see~\cite{MerT}
 for details about different models),  extrapolating our results about a single layer with
 two sequential phosporilations proved in Theorem~4.1 in~\cite{AGB1}.
 We assume the reader is familiar with the
 content of Section 2 in our companion paper~\cite{AGB1}. 
 
 We start with a couple of lemmas. Given a lattice configuration $\mathcal A$, we will denote by ${\rm Aff}(\mathcal A)$ the affine span
 of $\mathcal A$ consisting of of all points $\sum_{a\in \mathcal A} \lambda_a  \cdot a$ with
 $\lambda_a \in \Z$ for all $a \in \mathcal A$ and $\sum_{a \in \mathcal A} \lambda_a=0$.
  
\begin{lemma}\label{lem:samedim}
	Let $\mathcal{A} \subset \mathcal{A'} \subset \Z^d$ be finite point configurations, with 
	 ${\rm Aff}(\mathcal{A})={\rm Aff}(\mathcal{A'})=\Z^d$. Suppose that $\tau$ 
	  is a regular subdivision of $\mathcal{A}$. Then, there exists 
	  a regular subdivision $\tau'$ of  $\mathcal{A'}$ such that $\tau \subset \tau'$. 
	  Moreover,  we can choose a lifting function $h'$ inducing $\tau'$ such that $h:=h'|_{\mathcal{A}}$ induces $\tau$.
	
\end{lemma}	

\begin{proof}
	Let $h_{\tau}: \mathcal A \to \R$ be any lifting function inducing the subdivision $\tau$. 
	Let $h_{\mathcal A, \mathcal A'}: \mathcal A' \to \R$ be  any lifting function which is zero on
	$\mathcal A$ and positive otherwise. 
		 Extending $h_{\tau}$ by zero outside $\mathcal A$,
	 we get that for $\epsilon >0$ small enough the function $h':=h_{\mathcal A, \mathcal A'}+ \epsilon \cdot h_{\tau}$ induces a 
	 regular subdivision $\tau '$ of $\mathcal A'$ 
	containing the cells in $\tau$
	and $h:=h'|_{\mathcal{A}}=\epsilon \cdot h_{\tau}$ induces $\tau$.
\end{proof}

  \begin{lemma}\label{prop:diffdim}
  	Let $\mathcal{A} \subset \mathcal{A'} \subset \Z^{d'}$ be finite point configurations, 
  	with  ${\rm rk}{\rm Aff}(\mathcal{A})=d < {\rm rk}{\rm Aff}(\mathcal{A'})= d'$. 
  	Assume moreover that $\mathcal{A'} \setminus \mathcal{A}$
  	has cardinality $d'-d$. 
  	  	Suppose that $\tau$ is a regular subdivision (triangulation) of $\mathcal A$.
  	  	 For each $\sigma \in \tau$ define $\sigma'= \sigma
  	\cup (\mathcal{A'} \setminus \mathcal{A})$. Then, the collection $\tau':=\{\sigma', \sigma \in \tau\}$ 
  	defines a  regular subdivision (triangulation) of $\mathcal{A}'$. 
  Moreover, $\tau'$ can be induced by a lifting function $h'$ whose
  	restriction to $\mathcal{A}$ induces $\tau$.
  \end{lemma}

  \begin{proof}
  	Consider a point $a\in\mathcal{A'} \setminus \mathcal{A}$. Then $a$ is outside the hyperplane 
  	spanned by $\mathcal{A}$, that is, ${a}\cup \mathcal{A}$ is a 
  	pyramid over $\mathcal{A}$. It is known (see Observation $4.2.3$ in \cite{triangulations})
  	 that the collection $\{\sigma\cup {a}, \sigma \in \tau\}$ is a subdivision 
  	 of ${a}\cup \mathcal{A}$, and it is regular if and only if $\tau$ is regular. 
  	Then, we can see $\mathcal{A'}$ as an iterated pyramid over $\mathcal{A}$ 
  	and the lemma follows by applying successively the previous fact.
  \end{proof}
  
  Given a matrix $D \in \R^{d_D \times n_D}$ and $I \subset \{1, \dots,n_D\}$,
   we will denote by $D_I$ the submatrix of $D$ corresponding to the columns indexed by $I$. For $i\in \{1, \dots, n_D\}$,
  	$D(i)$ will denote the matrix obtained by removing the $i$-th column of $D$, and for $j \in \{1, \dots, d_D\}$,
  	 $D_j$ will denote the $j$-th row of  $D$.
  
  \begin{thm}\label{thm:extending}
  	Let  $d,d'\in \N$ with $d \le d'$. 
  	Let $\mathcal{A} \subset \mathcal{A''} \subset \Z^{d'}$ be finite point configurations, 
  with ${\rm rk}{\rm Aff}(\mathcal{A})=d $, ${\rm rk}{\rm Aff}(\mathcal{A''})=d'$.
Write	 $\mathcal{A}=\{a_1,\dots,a_n\}$,
  	$\mathcal{A''}=\mathcal{A}\cup \{a_{n+1},\dots,a_m\}$, with $ m \ge d' > n$. Set
	 $\mathcal{A'}=\mathcal{A}\cup \{a_{n+1},\dots,a_{n+d'-d}\}$ and assume that
${\rm rk}{\rm Aff} \mathcal{A'}=d'$.  
Let $\tau$ be a regular subdivision 
  	of $\mathcal{A}$ induced by a lifting function $h$, $\tau'$ the regular subdivision 
  	of $\mathcal{A'}$  obtained as in Lemma~\ref{prop:diffdim}, and $\tau''$ any regular 
  	subdivision of $\mathcal{A''}$ such that  $\tau'\subset\tau''$, 
  	induced by a lifting function $h''$, such that $h''$ restricted to $\mathcal{A}$ induces $\tau$.
  	
  	Let $f_1, \dots, f_d$ be polynomials with support in $\mathcal{A}$ and coefficient matrix
  	$C$ of rank $d$. Let $f''_1, \dots, f''_d, \dots,f''_{d'}$ be  polynomials with support $\mathcal{A''}$ and
  	coefficient matrix $C''$ of rank $d'$ of the form
  	\[
  	\begin{pmatrix}
  	C & 0  & D_1\\
  	M & B & D_2
  	\end{pmatrix},
  	\]
  	with $B\in\R^{(d'-d)\times (d'-d)}$ invertible.
  	Assume $\tau$ has $p$ $d$-simplices positively decorated
  	by $C$ and the determinants of the submatrices 
  	\[
  	\begin{pmatrix}
  	C_I \\
  	(B^{-1}M_I)_j
  	\end{pmatrix},
  	\]
  	have all the same sign as $(-1)^{d+i}\det(C_I(i))$, for each $i=1,\dots,d+1$, for each $j=1,\dots,d'-d$, and  for each subset $I\subset \{1, \dots,n\}$ 
  	which indexes a positively decorated simplex.
  	  	Then, there exists $t_0>0$, such that for $0<t<t_0$, the deformed system $f''_{1,t} = \dots = f''_{d',t}=0,$
  	where
  	\[f''_{i,t}(x) = \sum_{j=1}^m c''_{i,j} t^{h''(a_j)} x^{a_j},\]
  	has at least $p$ positive nondegenerate real roots.
  \end{thm}
  
  \begin{proof} 
The subdivision $\tau''$ can be obtained by Lemma~\ref{lem:samedim}. Note that the columns of $B$ correspond to the points in $\mathcal{A'}\setminus\mathcal{A}$. 

  Suppose that $\Delta\in \tau$ is  a $d$-simplex positively decorated by $C$. 
  Then $\Delta'=\Delta\cup\{a_{n+1},\dots,a_{n+d'-d}\}$ is a $d'$-simplex of 
  $\tau'\subset \tau''$. We will show that $\Delta'$ is positively decorated by $C''$.
  	
  	Suppose that $\Delta$ is indexed by the set $I=\{{i_1},\dots,{i_{d+1}}\}$ of $\{1, \dots,n\}$.  
  	We have to prove that the submatrix $C''_{I'}$ of $C''$ is positively spanning, 
  	with $I'=I\cup\{n+1,\dots,n+d'-d\}$. This is equivalent to prove that the matrix
  	\[G=  	
  	\begin{pmatrix}
  	\rm Id_{d} & 0 \\
  	0 & B^{-1} 
  	\end{pmatrix}
  	\begin{pmatrix}
  	C_I & 0 \\
  	M_I & B 
  	\end{pmatrix}=\begin{pmatrix}
  	C_I & 0 \\
  	B^{-1}M_I & \Id_{d'-d} 
  	\end{pmatrix},
  	\]
  	is positively spanning, as the property of being positively spanning remains invariant under multiplication by invertible matrices.
  	
  	We compute $(-1)^i \det( G(i))$ for $i=1,\dots,d'+1$. For $i=1,\dots,d+1$, we have:
  	\[(-1)^i \det( G(i))=
  	(-1)^i \det \begin{pmatrix}
  	C_I(i) & 0 \\
  	B^{-1}M_I(i) & \rm Id_{d'-d} 
  	\end{pmatrix}=(-1)^iC_I(i),
  	\]
  	which have all the same sign, because $\Delta$ is positively decorated by $C$.
  	Take now $i>d+1$. Let $j\in{1,\dots,d'-d}$ such that $d+1+j=i$. Moving the $i$-th row of $G$ 
  	to the row $d+2$ in $j-1$ interchanges of consecutive rows, we have:
  	\begin{align*}
  	(-1)^i \det( G(i))=&
  	(-1)^i \det \begin{pmatrix}
  	C_I & 0 \\
  	B^{-1}M_I & \rm Id_{d'-d}(j) 
  	\end{pmatrix}\\
  	=&(-1)^{d+1+j} (-1)^{j-1}\det \begin{pmatrix}
  	C_I & 0 \\
  	(B^{-1}M_I)_j  & \rm 0 \\
  	(B^{-1}M_I)[j]  & \rm Id_{d'-d -1}(j) 
  	\end{pmatrix}\\
  	=&(-1)^{d} \det \begin{pmatrix}
  	C_I \\
  	(B^{-1}M_I)_j
  	\end{pmatrix},
  	\end{align*}
  	where $(B^{-1}M_I)[j]$ denotes the submatrix of $B^{-1}M_I$ obtained by removing its $j$-th row. 
  	For each $j=1,\dots,d'-d$, this determinant has the same sign as 
  	$(-1)^iC_I(i)$, for each $j=1,\dots,d'-d$, by hypothesis.
  	  	Then, the simplex $\Delta'$ is positively decorated by $C''$. 
  	  	
  	  	We deduce that if $\tau$ 
  	has $p$ $d$-simplices positively decorated by $C$, then $\tau''$ has $p$ 
  	$d'$-simplices positively decorated by $C''$. Then, by Theorem~2.9 in~cite{AGB1} (from which we extracted \ref{th:BS}) , there exists 
  	$t_0>0$, such that for $0<t<t_0$, the system $f''_{1,t} = \dots = f''_{d',t}=0,$ has at least $p$ positive nondegenerate real roots.
  \end{proof}

  \begin{remark} \rm
  	The conditions that guarantee that the $p$ $d'$-simplices of $\tau''$ are positively 
  	decorated by $C''$ include the conditions such that the $p$ $d$-simplices of
  	$\tau$ are positively decorated by $C$, plus other conditions. In the cases of cascades of Goldbeter-Koshland loops that
  	we studied in Section~\ref{sec:3}, 
  	these other conditions are automatically fulfilled.
  \end{remark}

    \end{document}